\documentclass{article}

\usepackage{cite}

\usepackage{graphicx}
\usepackage{xcolor}
\usepackage{nicefrac}
\definecolor{ForestGreen}{rgb}{0.1333,0.5451,0.1333}
\definecolor{DarkRed}{rgb}{0.65,0,0}
\definecolor{Red}{rgb}{1,0,0}
\usepackage[linktocpage=true,
pagebackref=true,colorlinks,
linkcolor=DarkRed,citecolor=ForestGreen,
bookmarks,bookmarksopen,bookmarksnumbered]
{hyperref}

\usepackage{enumerate}

\usepackage[margin=1in]{geometry}                		% See geometry.pdf to learn the layout options. There are lots.
\usepackage{amssymb}
\usepackage{amsthm}
\usepackage{amsmath}
\usepackage{thm-restate}

\usepackage{algpseudocode}
\usepackage{algorithm}
\usepackage{float}
\usepackage{tablefootnote}
\usepackage{longtable}

\usepackage{xspace}
\usepackage[colorinlistoftodos,textsize=tiny,textwidth=2cm,color=red!25!white,obeyFinal]{todonotes}
\usepackage{cleveref}

\newtheorem{theorem}{Theorem}[section]
 
 \newtheorem{lemma}[theorem]{Lemma}
\newtheorem{observation}[theorem]{Observation}

\newtheorem{definition}[theorem]{Definition}

\newtheorem{conjecture}[theorem]{Conjecture}

\newcommand{\OO}{\mathcal{O}}

\newcommand{\APX}{1.546}
\newcommand{\APXs}{\alpha_{alg}}

\newcommand{\E}{\mathbb{E}}
\newcommand{\N}{\mathbb{N}}
\newcommand{\Z}{\mathbb{Z}}
\newcommand{\R}{\mathbb{R}}
\newcommand{\mue}{\mu_{\delta}}
\newcommand{\Inf}{\mathsf{Inf}}
\newcommand{\Valug}{\mathsf{Val_{UG}}}
\newcommand{\Optug}{\mathsf{Opt_{UG}}}
\newcommand{\calV}{\mathcal{V}}
\newcommand{\calU}{\mathcal{U}}

\newcommand{\eps}{\epsilon}
\newcommand{\KM}{\textsc{$k$-Median}\xspace}
\newcommand{\KME}{\textsc{$k$-Median}\xspace}
\newcommand{\CKME}{\textsc{Capacitated $k$-Median}\xspace}

\newcommand{\ksupplier}{\textsc{$k$-Supplier}\xspace}
\newcommand{\kcenter}{\textsc{$k$-Center}\xspace}
\newcommand{\FL}{\textsc{Facility Location}\xspace}
\newcommand{\PP}{\mathbf{P}}
\newcommand{\NP}{\mathbf{NP}}
\newcommand{\MKC}{\textsc{Maximum $k$-Coverage}\xspace}
\newcommand{\numalgo}{1.546}
\newcommand{\numhard}{1.416}

\newcommand{\nf}{\nicefrac}

\title{Separating \KME from the Supplier Version}
\author{Aditya Anand \qquad Euiwoong Lee \\ University of Michigan, Ann Arbor}
\date{}

\begin{document}
    
\maketitle
\pagenumbering{arabic}
\begin{abstract}

Given a metric space $(V, d)$ along with an integer $k$, the \KME problem asks to open $k$ centers $C \subseteq V$ to minimize $\sum_{v \in V} d(v, C)$, where $d(v, C) := \min_{c \in C} d(v, c)$. While the best-known approximation ratio 2.613 holds for the more general {\em supplier version} where an additional set $F \subseteq V$ is given with the restriction $C \subseteq F$, the best known hardness for these two versions are $1+1/e \approx 1.36$ and $1+2/e \approx 1.73$ respectively, using the same reduction from \MKC. We prove the following two results separating them.

\begin{enumerate}
    \item We give a \numalgo-parameterized approximation algorithm that runs in time $f(k) n^{O(1)}$. Since $1+2/e$ is proved to be the optimal approximation ratio for the supplier version in the parameterized setting, this result separates the original \KME from the supplier version. 
    
    \item We prove a \numhard-hardness for polynomial-time algorithms assuming the Unique Games Conjecture. This is achieved via a new fine-grained hardness of \MKC for small set sizes. 
\end{enumerate}

Our upper bound and lower bound are derived from almost the same expression, with the only difference coming from the well-known separation between the powers of LP and SDP on (hypergraph) vertex cover. 
% \keywords{Approximation algorithms \and Clustering \and Parameterized approximation \and Optimization \and Hardness of approximation}
\end{abstract}
\newpage

\section{Introduction}
\KME is perhaps the most well-studied clustering objective in finite general metrics. Given a (semi-)metric space $(V, d)$ along with an integer $k$, the goal is to open $k$ centers $C \subseteq V$ to minimize the objective function $\sum_{v \in V} d(v, C)$, where $d(v, C) := \min_{c \in C} d(v, c)$.
It has been the subject of numerous papers introducing diverse algorithmic techniques, including filtering~\cite{lin1992approximation, lin1992approximations}, metric embedding~\cite{bartal1998approximating, charikar1998rounding}, local search~\cite{arya2001local}, primal LP rounding~\cite{charikar1999constant, charikar2012dependent}, primal-dual~\cite{jain2001approximation}, greedy with dual fitting analysis~\cite{charikar1999improved, jain2003greedy}, and bipoint rounding~\cite{li2013approximating, byrka2017improved}. 
The current best approximation ratio is slightly lower than 2.613~\cite{cohen2023breaching, gowda2023improved}, achieved by the combination of greedy search analyzed dual fitting method and bipoint rounding. 

One can define a slightly more general version of \KME by putting restrictions on which points can become centers; the input specifies a partition of $V$ into {\em candidate centers} $F \subseteq V$ and {\em clients} $L \subseteq V$, and the goal is to choose $C \subseteq F$ with $|C| = k$ to minimize $\sum_{v \in L} d(v, C)$. Let us call this version the {\em supplier version} in this paper. This distinction between the original version and the supplier version is apparent in a similar clustering objective \kcenter (where the objective function is $\max_{v \in C} d(v, C)$), whose supplier version has been called \ksupplier and sometimes studied separately. \kcenter and \ksupplier admit $2$ and $3$ approximation algorithms respectively, which are optimal assuming $\PP \neq \NP$, so we know there is a strict separation between these two. 

For \KME, this distinction has not been emphasized as much. 
Earlier papers define \KME as the original version~\cite{lin1992approximation, lin1992approximations, bartal1998approximating, charikar1998rounding, charikar1999constant}, but most papers in this century define \KME as the supplier version. 
One reason for this transformation might be the influence of techniques coming from the \FL problem. 
As the name suggests, \FL originates from the planning perspective (just like \ksupplier), so it is natural to assume that $V$ is partitioned into the set of clients $L$ and set of potential facility sites $F$ with opening costs $f : F \to \R^+$ and the goal is to open $C \subseteq F$ (without any restriction on $|C|$) to minimize $f(C) + \sum_{v \in L} d(v, C)$. 
But once we interpret \KME in the context of geometric clustering, such a restriction seems unnecessary, and it is more natural to study the original version directly. 

Then the question is: will the original version of \KME exhibit better approximability than the supplier version, just like \kcenter and \ksupplier? Algorithmically, no such separation is known, as many of the current algorithmic techniques, including the ones giving the current best approximation ratio, rely on the connection between \KME and \FL; the best approximation ratio for \KME still holds for the supplier version. 
However, these two problems show different behaviors in terms of hardness.
The reduction from \MKC to \FL by Guha and Khuller~\cite{guha1999greedy}, adapted to \KME by Jain et al.~\cite{jain2003greedy}, shows that the supplier version of \KME is NP-hard to approximate within a factor $(1+\nf2e) \approx 1.73$. The same reduction yields only $(1+\nf1e) \approx 1.36$ for the original \KME. This separation in the current hardness, together with the separation between \kcenter and \ksupplier, suggest that the best approximation ratio achieved by a polynomial-time algorithm for the original \KME might be strictly smaller than that of the supplier version. Of course, formally showing such a separation requires a polynomial-time algorithm for the original version with the approximation ratio strictly less than $1+\nf2e$, which seems hard to achieve now given the big gap between the upper and lower bounds for both  versions. 

A relatively new direction of {\em parameterized approximation}, which studies algorithms running in time $f(k) \cdot n^{O(1)}$ for any computable function $f(k)$, 
might be a good proxy for the polynomial-time approximation. In fact, Cohen-Addad et al.~\cite{cohen2019tight} showed that there exists a parameterized approximation algorithm for the supplier version that guarantees a $(1+\nf2e+\eps)$-approximation for any $\eps > 0$. Moreover, this additional running time does not seem to cross the NP-hardness boundary; assuming the Gap Exponential Time Hypothesis (Gap-ETH), they show that no FPT algorithm can achieve an approximation ratio $(1+\nf2e - \eps)$ for any $\eps > 0$. Many subsequent (and concurrent) papers studied parameterized approximability of many variants of \KME and obtained improved approximation ratios over polynomial-time algorithms~\cite{cohen2019fixed, adamczyk2019constant, goyal2020fpt, feng2020unified, inamdar2020capacitated, xu2020constant, bandyapadhyay2022fpt, jaiswal2023clustering, goyal2023tight, agrawal2023clustering, bandyapadhyay2023fpt}, but for quite a few of these problems, including \KME and \CKME, no parameterized approximation algorithm achieved a result impossible for their polynomial-time counterparts. Therefore, apart from its practical benefit when $k$ is small, the study of parameterized approximation might yield a meanigful prediction for polynomial-time approximability.

Our main algorithmic result is the following parameterized approximation algorithm that strictly separates the original and supplier versions of \KME. To the best of our knowledge, it is the first separation between the original and supplier versions of any variant of \KME. 

\begin{restatable}{theorem}{algomain}
For any $\eps > 0$, 
there is an algorithm for \KME that runs in time $f(k, \eps)\cdot n^{O(1)}$ and guarantees an approximation ratio of 
\[
\min_{p \in [0, 1]} 
\max_{d \in \Z_{\geq 1}} 
\bigg[ 
1 + \bigg( 1 - \frac{1-p}{d} \bigg) ^d
+ \bigg( 1 -  \frac{pd + (1 - p)}{d}\bigg)^d + \eps
\bigg] \approx \numalgo + \eps,
\]
where the optimal value is achieved with $p^* := \nf{(10 - 6\sqrt{2}) }{7}\approx 0.22$ and $d^* := 3$. 
\label{thm:algo}
\end{restatable}
Our main hardness result is the following theorem against {\em polynomial-time} algorithms assuming the Unique Games Conjecture (UGC)~\cite{Khot02}, improving the best known lower bound of $(1+\nf1e) \approx 1.36$~\cite{guha1999greedy, jain2003greedy}. 

\begin{restatable}{theorem}{hardnessmain}
Assuming the Unique Games Conjecture, for any $\eps > 0$, \KM is NP-hard to approximate within a factor 
\[
\max_{d \in \Z_{\geq 3}} 
\min_{p \in [0, 1]} 
\bigg[
1 + \bigg( 1 - \frac{1-p}{d-1} \bigg) ^d
+ \bigg( 1 - \min(1, \frac{pd + (1 - p)}{d-1})\bigg)^d  - \eps \bigg] \approx \numhard - \eps.
\]
\label{thm:hardness}
\end{restatable}

Somewhat surprisingly, this expression almost exactly matches the algorithmic guarantee except (1) denominators being $d-1$ instead of $d$ (which causes the $\min(1, .)$) and (2) $\max_d \min_p$ vs $\min_p \max_d$.

The difference between $d-1$ and $d$ comes from the well-known separation between the powers of LP and SDP on (hypergraph) vertex cover, elaborated more in the overview below.  Furthermore, we show that the max-min and the min-max values of the function in \Cref{thm:algo} coincide (\Cref{lemma:optimize}), so these two results might point towards the optimal approximability of \KME, both for polynomial-time and parameterized algorithms. 

Our hardness result proceeds via, and hence shows, more fine grained hardness for \MKC. Given an instance of \MKC with universe size $n$ and a constant $d \geq 3$ so that each element appears in exactly $d$ sets, we show that assuming the Unique Games Conjecture, it is NP-Hard to distinguish the YES-case where $k$ sets cover all elements, and the NO-case where for any $\alpha \in [0,1]$, $\alpha k$ sets cover at most $1 - (1 - \frac{\alpha}{d-1})^d$ fraction of the elements where $k = \frac{n}{d-1}    $(see~\Cref{thm:hvc}). This is a new fine-grained version of  the $(1 - \frac{1}{e})$ hardness of~\cite{feige1998threshold} when $\alpha = 1$ for small $d$.

\subsection{Overview}
In this overview, we explain our intuition behind these two results and how they are related. For the sake of simplicity, let us ignore the arbitrarily small constant $\eps > 0$ in the overview.

\paragraph{Hardness.}
Our improved hardness, which eventually guided our algorithm, comes from observing the classical reduction of Guha and Khuller~\cite{guha1999greedy} and Jain et al.~\cite{jain2003greedy}. It is a reduction from the famous \MKC problem; given a hypergraph $H = (V_H, E_H)$, and $k \in \N$, choose $S \subseteq V_H$ of $k$ vertices that intersects the most number of hyperedges.\footnote{Typically, it is stated in terms of the dual set system where the input is a set system, and the goal is to choose $k$ sets whose union size is maximized.} 
The classical result of Feige~\cite{feige1998threshold} shows that it is NP-hard to distinguish between the (YES) case where $k$ vertices intersect all hyperedges and the (NO) case where any choice of $k$ vertices intersect at most $(1-\nf1e)\approx 0.63$ fraction of the hyperedges. 

Given a hardness instance $(H, k)$ for \MKC, the reduction to \KME is simply to construct the vertex-hyperedge incidence graph $G = (V_G, E_G)$ where $V_G = V_H \cup E_H$ and a pair $(v, e) \in V_H \times E_H$ is in $E_G$ if $v \in e$ in $H$. 
In the supplier version, one can simply let $F = V_H$ and $L = E_H$. Given $C \subseteq F$, for each $e \in L$, the distance $d(e, C)$ is $1$ if it contains a vertex in $C$ and at least $3$ otherwise, thanks to the bipartiteness of $G$. Therefore, the average distance for clients in $L$ becomes $1$ in the YES case and at least $1 \cdot (1-\nf1e) + 3 \cdot (\nf1e) = 1 + \nf2e$ in the NO case.
% \anote{Is Feige NP-Hardness under this strong requirement? Or is it some weaker assumption, should check.}

The same construction works for the original version without $F$ and $L$ with two differences. The first difference is that the objective function sums over points in $V_H$ as well as $E_H$, but this can be handled by duplicating many points for each $e \in E_H$. The bigger issue is the fact that $C$ might contain points from $E_H$, which implies that $d(e, C)$ can be possibly $2$ even when $C$ does not contain a vertex from $e$. Therefore, the hardness factor weakens to $1 \cdot (1-\nf1e) + 2 \cdot (\nf1e) = 1 + \nf1e$. 

The natural question to ask to improve the hardness is then: how much does placing centers at $E_H$ help cover other points in $E_H$? 
%Suppose every hyperedge in $H$ has size $d$, so that in the YES case, $k := |V_H| / d$ vertices intersect (almost) all hyperedges. 
Suppose that a solution $C$ contains $(1-p) k$ centers from $V_H$ and $p k$ centers from $E_H$ for some $p \in [0, 1]$. 
Compared to the solution that puts all centers in $V_H$, putting some centers at $E_H$ will hurt the ability to cover points in $E_H$ at distance $1$ (i.e., for $e \in E_H$, $d(e, C) = 1$ only if there exists $v \in (C \cap e)$), but it will help covering points in $E_H$ at distance $2$ (i.e., for $e \in E_H$, $d(e, C) \leq 2$ if there exists $f \in (C \cap E_H)$ with $e \cap f \neq \emptyset$).
To bound the effect of the latter, one natural and conservative observation is, when the hypergraph $H$ has uniformity $d$ (i.e., each $e \in E_H$ has $|e| = d$), letting $C' = (V_H \cap C) \cup (\cup_{e \in (E_H \cap C)} e)$, a hyperedge $e \in E_H$ has $d(e, C) \geq 3$ if it does not intersect $C'$. 
Since each $e \in C \cap E_H$ generates at most $d$ new points in $C'$, we can easily see $|C'| \leq ((1-p) + d p )k$. 
Therefore, it seems we need a good hardness of approximation for $\MKC$ with small $d$ in order to go beyond $1+\nf1e$. 
Quantitatively, for a $d$-uniform hypergraph $H$, the ideal hardness one can imagine is the following.

\begin{itemize}
    \item In the YES case, $k := \nf{|V_H|}{d}$ vertices intersect (almost) all the hyperedges. For \KME, there exist $k$ centers that cover almost all points at distance $1$. 
    
    \item In the NO case, $H$ is like a {\em random hypergraph}; for any $\gamma \in [0, 1]$, any choice of $\gamma |V_H|$ vertices intersect at most $1 - (1 - \gamma)^d$ hyperedges. In other words, for fixed $\gamma$, the best solution is essentially to pick each vertex with probability $\gamma$. 

    For \KME, for any solution $C$ with $|C \cap V_H| = (1-p) k$ and $|C \cap E_H| = p k$, we use the above guarantee twice to $(C \cap V_H)$ and $C'$ to conclude that (1) at most a $1 - (1 - \nf{(1-p)}{d})^d$ fraction of hyperedges can be covered by distance at most $1$ and (2) at most $1 - (1 - \nf{(1-p + dp)}{d})^d$ fraction of hyperedges can be covered by distance at most $2$. 
\end{itemize}

Such an ideal hardness turns out to be impossible to prove, but one can get close. For any integer $d \geq 3$, one can prove the above ideal hardness for $d$-uniform hypergraphs with the only difference being $k := \nf{|V_H|}{d-1}$ instead of $\nf{|V_H|}{d}$. This means that the denominator $d$ is replaced by $d-1$ in the expressions $1 - (1 - \nf{(1-p)}{d})^d$ and $1 - (1 - \nf{(1-p + dp)}{d})^d$ above, which yields the hardness result of \Cref{thm:hardness}. If we had proved the ideal hardness above, then the hardness ratio would have exactly matched the algorithm of \Cref{thm:algo}. We elaborate more about this gap in the discussion below.

\paragraph{Algorithms.}
Our algorithm follows the framework of Cohen-Addad et al.~\cite{cohen2019tight} for the supplier version. 
One can compute a {\em coreset} $P \subseteq V$, $|P| = O(k \log n)$ with weight function $w : P \to \R^+$ so that for any choice of $k$ centers, the objective function for points in $V$ is almost the same as that for weighted points in $P$. Therefore, one can focus on $P$ as the set of points. 
Let $P^*_1, \dots, P^*_k$ be the partition of $P$ with respect to the optimal clustering. For each $i \in [k]$, let $c^*_i \in V$ be the optimal center corresponding to $P^*_i$, and $\ell_i = \mathrm{argmin}_{\ell \in P^*_i} d(\ell, c^*_i)$ be the {\em leader} of $P^*_i$. Guessing $\ell_1, \dots, \ell_k$ and (approximately) guessing $d(\ell_1, c^*_1), \dots, d(\ell_k, c^*_k)$
takes a parameterized time $O(k \log n)^k$. 
Then, for each $i \in [k]$, let $C_i$ be the points that have distance (almost) equal to $d(\ell_i, c^*_i)$. By definition $c^*_i \in C_i$.

We solve the standard LP relaxation for \KME, with only the additional constraint that we fractionally open exactly one center from each $C_i$. (A standard trick can make sure that $C_i$'s are disjoint.)
The rounding algorithm opens exactly one center from $C_i$ (different clusters are independent) according to LP values.
Fix a point $v \in P$. A standard analysis for \KME shows that with probability at least $1 - \nf1e$, a center fractionally connected to $v$ will open, and the expected distance from $v$ to the closest open center conditioned on this event is at most its contribution to the LP objective. In the other event, if $v \in P^*_i$ for some $i \in [k]$, since one center $c \in C_i$ is open and 
\[
d(v, c) \leq 
d(v, c^*_i) +  
d(\ell_i, c^*_i) +
d(\ell_i, c),
\]
the fact that $d(v, c^*_i) \geq d(\ell_i, c^*_i)$ (by definition of $\ell_i$) and 
$d(\ell_i, c^*_i) \approx d(\ell_i, c)$ (by definition of $C_i$) ensures that this distance is almost at most $3d(v, c^*_i)$. Combining these two events shows that the total expected distance is at most $(1-\nf1e)LP + (\nf3e)OPT \leq (1+\nf2e)OPT$. 

In order to obtain a possibly better result over the supplier version, using the power that we can open centers anywhere, for some fixed $p \in [0, 1]$, we do the following. For each $i \in [k]$, we simply open the leader $\ell_i$ as a center with probability $p$, and with the remaining probability $(1-p)$, open exactly one center from $C_i$ as usual. 
For the analysis, let us fix $v \in P_i$ and consider its distance to the closest center in this rounding strategy. For simplicity, assume $v$ is fractionally connected to $c_1, \dots, c_d$ with fraction $\nf1d$ each, where $c_i \in C_i$ for $i \in [d]$. Let us also assume that $d(v,c_i) = 1$ for each $i \in [d]$, and $d(\ell_i,c_i) = 1$ for each $i \in [d]$. Observe that
$c_i$ opens with probability $\nf{(1-p)}{d}$, $\ell_i$ opens with probability $p$, and some other center in $C_i$ opens with probability $\nf{(1-p)(d-1)}{d}$. If at least one $c_i$, $i \in [d]$ is opened, then $v$'s connection cost is at most $1$. If at least one $\ell_i, i \in [d]$ is open, then $v$'s connection cost is at most $2$.
It follows that the probability that no center at distance $1$ from $v$ opens is exactly $(1 - \nf{(1-p)}{d})^d$, and the probability that no center at distance $\leq 2$ from $v$ opens is exactly $(\nf{(1-p)(d-1)}{d})^d = (1 - \nf{(1-p + pd)}{d})^d$.
Note they exactly match the analysis for the (ideal) hardness; the first probability $(1 - \nf{(1-p)}{d})^d$ corresponds to the case where we choose random $(1-p)k$ vertices in the hypergraph, and the second probability $(1 - \nf{(1-p + pd)}{d})^d$ corresponds to the case where we choose random $(1-p + pd)k$ vertices due to picking $pk$ hyperedges!

Of course, the above intuition already assumes various regularities involving the LP values and distances. We show that such regular cases are indeed the worst case in terms of approximation ratio. The analysis involves a series of factor-revealing programs that try to find the worst-case configuration for a fixed point, which reveal simplifying structural properties of the optimal configuration via rank and convexity arguments.

\paragraph{Discussion and Future Work.}
The first question from this work might be: where is the difference between $d$ and $d-1$ coming from, and how can we close this gap? 
As we saw, if we had been able to prove the ``ideal'' hardness, that is, the inability to distinguish a $d$-uniform hypergraph with vertex cover size $\nf{|V_H|}{d}$ and a random $d$-uniform hypergraph, this gap would not exist.

If we just want to {\em fool} the LP, since every $d$-uniform hypergraph admits a fractional covering of all hyperedges with a $\nf1d$ fraction of vertices, we believe that a random $d$-uniform hypergraph will exhibit such a property, giving an LP (even Sherali-Adams) gap with the exact same factor as the upper bound. Indeed, when $d = 2$, Charikar, Makarychev, and Makarychev~\cite{charikar2009integrality} showed that even $O(n^{\delta})$-levels of Sherali-Adams cannot distinguish a random graph and a nearly-bipartite graph.

However, for $d$-uniform hypergraphs with small $d$, SDPs are expected to strictly outperform the LP, and this is evident when $d=2$ where the LP (even Sherali-Adams) guarantees at most $\nf34$-approximation for \MKC while the SDP guarantees a $0.929$-approximation~\cite{manurangsi2019note}. Our Unique Games Hardness is based on a standard way to construct an SDP gap based on pairwise independence, which is essentially from H\aa stad's seminal result for \textsc{Max 3LIN} and \textsc{Max 3SAT} when $d = 3$~\cite{haastad2001some}. It is reasonable to expect that SDPs give a strictly better guarantee for \KME instances constructed from (the vertex-hyperedge incidence graph of) $d$-uniform hypergraphs, but reducing the general \KME instance to such a case seems a significant technical challenge as our analysis relies on simple properties of the LP.

A more fundamental limitation of our work is that the upper bound is given by parameterized algorithms and the lower bound is only against polynomial-time algorithms. The biggest open question is to prove an optimal approximability for either polynomial-time or parameterized algorithms. Currently, it might be the case that the optimal thresholds for these two classes of algorithms coincide just like the supplier version.

Improving the best polynomial-time approximation ratio for any version of \KME has been studied intensively. Matching the $(1+\nf2e)$-hardness for the supplier version or going below for the original version is a long-standing open problem, but we believe that giving an algorithm for the original version with the approximation ratio strictly better than the current known 2.613~\cite{cohen2023breaching, gowda2023improved} for the supplier version is a meaningful step. Such an algorithm is likely to yield insight into how to take advantage of opening centers anywhere in the context of primal-dual or dual-fitting analysis, whereas our improved parameterized algorithm uses a primal LP rounding algorithm.

From the parameterized hardness perspective, we first remark that \MKC in $d$-uniform hypergraphs with constant $d$ (or even $d = f(k)$) admits a parameterized approximation scheme~\cite{badanidiyuru2012approximating, jain2023parameterized}. However, there might be more sophisticated ways to understand the effect of opening centers anywhere. A better understanding of the hardness of \textsc{Label Cover} in the parameterized setting, combined with Feige's reduction to \MKC, might be a way to achieve a strictly improvement over $(1 + \nf1e)$.

\section{Approximation algorithm}

In this section, we prove our main algorithmic result \Cref{thm:algo}, restated below. 

\algomain*

\Cref{subsec:algo} presents our algorithm. 
\Cref{subsec:analysis-setup} introduces the setup for our analysis, which tries to understand the worst-case configuration for a fixed point via a series of factor-revealing programs.
\Cref{subsec:s},~\Cref{subsec:mu} and~\Cref{subsec:final} consider different such programs and deduce that the optimal solutions exhibit simple structures, which leads to our final bound.

\subsection{Algorithm}
\label{subsec:algo}
We begin by computing a \emph{coreset}. For any choice of centers $C \subseteq V$ and a set of weighted points $V' \subseteq V$ where point $v \in V'$ has weight $w_v$, let $cost(V', C)$ denote the cost of assigning the points $V'$ to the centers $C$ so as to minimize the total weighted connection cost. Concretely, for every assignment function $f: V' \rightarrow C$, define the cost of $f$ as $\sum_{v \in V'} w_v d(v,f(v))$. Then $cost(V', C)$ is the minimum cost among all such assignments $f$.
\begin{definition}
Given a \KM instance $(V,d,k)$ and $\eps > 0$, an $\eps$-coreset $V'$ is a set of points with a weight $w_v$ for each $v \in V'$, so that for any choice of centers $C \subseteq V$ with $|C| = k$, $cost(V', C) \in (1-\epsilon, 1+ \epsilon) cost(V, C)$. 
\end{definition}
\begin{theorem}[\cite{fl11}]\label{thm:coreset}
There is a polynomial time algorithm that when given a \KM instance and $\eps > 0$, computes an $\eps$-coreset $V' \subseteq V$ with $|V'| \leq \OO(\nf{k\log n}{\epsilon^2})$.
\end{theorem}
We start by considering a natural LP relaxation and proceed in a similar manner to the algorithm of~\cite{cgkll19}. We first  compute a coreset $V'$ for the point set $V$ of size $\OO(\frac{k\log n}{\epsilon^2})$, using~\Cref{thm:coreset}. It follows that for any choice of $k$ centers $C \subseteq V$, we have $cost(V',C) \in (1 - \epsilon, 1 + \epsilon) cost(V,C)$. Henceforth we focus on minimizing $cost(V',C)$.

For the optimal partition $\{V_1^*, V_2^*, \ldots, V_k^*\}$ of $V'$ and corresponding choice of centers $\{c_1^*, c_2^*, \ldots, c_k^*\}$, we guess for each $i \in [k]$ the leaders $\ell_i$, which are the points in $V_i^*$ closest to $c_i^*$ for each $i \in [k]$. Since the coreset guarantees $|V'| \leq 
\OO(\frac{k\log n}{\epsilon^2})$, this can be accomplished in time $\OO((\frac{k \log n}{\epsilon^2})^k)$. We also guess the distances $R_i^*$ of the leaders $\ell_i$ to the corresponding center $c_i^*$ for each $i \in [k]$, rounded down to the nearest power of $(1 + \epsilon)$. By the argument used in~\cite{cgkll19}, we can assume without loss of generality that the distance between any two points of $C$ is at least $1$ and at most $n^{\OO(1)}$: this means there are only $\OO({\log (\frac{n}{\epsilon})})$ choices for the distances $R_i^*$, which in turn means that this step can be accomplished in time $\OO(({\log \frac{n}{\epsilon}})^k)$. It can be easily shown that $\OO(({\log \frac{n}{\epsilon}})^k)$ and $\OO((\frac{k \log n}{\epsilon^2})^k)$ can be bounded above by $f(k, \frac{1}{\epsilon}) n^{\OO(1)}$ as follows. First, note that it is enough to upper bound $({\log \frac{n}{\epsilon}})^{\OO(k)}$. If $k > \frac{\log \frac{n}{\epsilon}}{\log \log \frac{n}{\epsilon}}$, then this expression is upper bounded by $2^{\OO(k \log k)}$. Otherwise, $k < \frac{\log \frac{n}{\epsilon}}{\log \log \frac{n}{\epsilon}}$, and we get an upper bound of $(\log \frac{n}{\epsilon})^{\OO \left(\frac{\log \frac{n}{\epsilon}}{\log \log \frac{n}{\epsilon}}\right)} \leq (\frac{n}{\epsilon})^{\OO(1)}$.

Once we guess these quantities, let $C_i$ be the subset of $V$ which are at distance at most $R_i^*(1 + \epsilon)$ from $\ell_i$. Notice that we must have $c_i^* \in C_i$. For purely technical reasons, by adding a copy of $\ell_i$ in $C_i$, we will assume that $\ell_i \notin C_i$. Also by making copies of points, we will assume that $C_1, \dots, C_{\ell}$ are disjoint.

Consider the following natural LP relaxation for the problem.

\begin{center}
\noindent\fbox{
  \begin{minipage}{0.95\textwidth}
  \begin{align}
  \min \quad & \sum_{v \in V'} w_v \sum_{c \in V} d(v,c) x_{v, c} \nonumber \\
  \mbox{s.t.} \quad & x_{v,c} \leq y_c && \forall v \in V', c \in V \nonumber \\
  & \sum_{c \in C_i} y_c = 1  && \forall i \in [k] \nonumber \\
  &\sum_{c \in V} x_{v,c} = 1 && \forall v \in V' \nonumber \\
  & x_{v, c} = 0 && \forall v \in V' , i \in [k], c \in C_i \mbox{ such that } d(v,c) < d(\ell_i,c) \label{eq:tooclose}
  \end{align}
  \end{minipage}
       }
\end{center}

Clearly, this LP is a relaxation as witnessed by the canonical integral solution where one can assign each point in $V_i^*$ to (a copy of) $c_i^* \in C_i$.

We solve the LP to obtain a (fractional) optimal solution $(\mathbf{x},\mathbf{y})$ for the LP - we will also assume that $x_{v, c} = y_c$ for each $c \in V$ and each point $v \in V'$ with $x_{v, c} \neq 0$. This can be accomplished using a fairly straightforward construction by   ``splitting centers'' which we describe below. Similar constructions have been used in previous work for other related problems.

Pick a point $c \in V$ for which there is a point $v$ such that $x_{v, c} \neq 0$ and $x_{v, c} \neq y_c$. Let $v_1, v_2 \ldots v_t$ be the points $v$ with $x_{v, c} \neq 0$, taken in non-decreasing order of $x_{v, c}$. Create $t + 1$ many copies $c_1, c_2 \ldots c_{t+1}$ of the point $c$ (with all distances from other points the same as the original point $c$). Set $y_{c_1} = x_{v_1, c}$. Next, set $y_{c_2} = x_{v_2, c} - x_{v_1, c}$, and so on, so that $y_{c_j} = x_{v_j, c} - x_{v_{j-1}, c}$ for all $2 \leq j \leq t$. Set $y_{c_{t+1}} = 1 - \sum_{j = 1}^t y_{c_j}$. Finally, for each point $v_i$, $i \in [t]$, we set $x_{v_i, c_j} = y_{c_j}$ for $j \in \{1,2 \ldots i\}$. It is clear that the construction preserves feasibility and optimality, and $x_{v, c} = y_c$ whenever $x_{v, c} \neq 0$.

Our rounding algorithm is very simple: we pick a threshold $p \in [0,1]$ to be determined later by analysis, and for each $i \in [k]$, we pick $\ell_i$ into $C$ with probability $p$, and with remaining probability $1-p$, we choose exactly one center $c$ from $C_i$ from the distribution given by the LP: more concretely, each center $t \in C_i$ is picked with probability $y_t$. 

\begin{algorithm}[!htbp]
\begin{algorithmic}[!htbp]
\caption{Rounding algorithm}
\label{alg:rounding}
\State \textbf{Input:} \KM{} instance $(V, d, k)$ with coreset $V'$, the sets $C_1, C_2 \ldots C_k$ and leaders $\ell_1, \ell_2 \ldots \ell_k$. A parameter $p \in [0,1]$, and the optimal LP solution $(\mathbf{x},\mathbf{y})$.
\State \textbf{Output:} An approximate solution $C \subseteq V$ with $|C| = k$.
\For{$i \in [k]$}
\State With probability $p$ pick $\ell_i$. 
\State With remaining probability $1-p$, pick exactly one center $t \in C_i$ with probability $y_t$. 
 \EndFor
 \end{algorithmic}

\end{algorithm}

Observe that for every $i \in [k]$, the algorithm opens either a center in $C_i$, or the leader $\ell_i$.

\subsection{Analysis: Casting as factor-revealing problems}
\label{subsec:analysis-setup}
Once we decide on the centers, each point $v \in V'$ can be assigned to the closest open center. 
For the sake of the analysis, we consider the following specific assignment of the points to the open centers; the cost only increases by doing so. 
Henceforth, we fix $v \in V'$. Let $I_v \subseteq [k]$ be the set of indices $i$ for which there is a center $c \in C_i$ with $x_{v, c} \neq 0$, and suppose $|I_v| = \ell$. For the sake of simplicity and clarity of exposition, by re-numbering indices, let us assume that $I_v = [\ell]$. 

For each $i \in [\ell]$, define the \emph{flow} to cluster $i$ as $\mu_i = \sum_{c \in C_i} x_{v, c}$. Let $s_i$ be the weighted average distance to a center $c \in C_i$, where center $c$ is given weight $x_{v, c}$, so that $s_i = \nf{(\sum_{c \in C_i} x_{v, c} d(v,c))}{\mu_i}$.  Assume without loss of generality that $s_i \leq s_j$ whenever $i \leq j$ for $i,j \in [\ell]$. Note that $s_i$ is the \emph{expected} cost of assigning $v$ to a random center $c \in C_i$ according to the probabilities $x_{v, c}$.~\Cref{alg:asgn} describes our assignment algorithm.

\begin{algorithm}[!htbp]
\begin{algorithmic}[!htbp]
\caption{Assignment algorithm}
\label{alg:asgn}
\State \textbf{Input:} \KM{} instance with coreset $V'$ and centers $C$ opened by~\Cref{alg:rounding}, a point $v \in V'$, and the set $I_v = [\ell]$.  
\State \textbf{Output:} An assignment of $v$ to a center of $C$.

%\State Partition the indices $[\ell]$ based on their distance $s_i$ from $v$. 
\State Let $A$ be the set of indices $i \in [\ell]$ with $s_i \leq 1.5s_{1}$.
%\State Let $B$ be the set of indices $i \in [\ell]$ with $s_i \geq 1.5s_1$.
\State Let $D$ be the set of centers $c$ satisfying $x_{v, c} \neq 0$, with $c \in C_i$ so that $s_i \leq 3s_1(1 + \epsilon)$ for some $i \in [\ell]$. 
\If{$D \cap C \neq \emptyset$}
\State Assign $v$ to the center $c \in C_{i^*}$, where $i^*$ is the first index with $C_{i^*} \cap D \cap C \neq \emptyset$. 
\ElsIf{there is a $j \in A$ such that $\ell_j$ is open}
\State Assign $v$ to $\ell_{j^*}$ where $j^*$ is the first index in $A$ such that $\ell_{j^*}$ is open. 
\Else 
\State Assign $v$ to any open center in $C_1$.
 \EndIf
 \end{algorithmic}

\end{algorithm}

First, we have the following lemma which shows that $v$ will be always assigned to a relatively close center. 

\begin{lemma}\label{lemma:dist}
For any $j \in [\ell]$, we must have $d(v, \ell_j) \leq 2s_j$. Consequently, $v$ will be always assigned an open center at distance at most $3s_{1}(1 + \epsilon)$.
\end{lemma}

\begin{proof}
Fix any index $j \in [\ell]$. Let $c'_{j} \in C_{j}$ be a point in $C_j$ which minimizes $d(v, c)$ among all points $c \in C_j$ with $x_{v, c} \neq 0$. Clearly, $d(v,c_j') \leq s_j$. By the triangle inequality, we have $d(v, \ell_j) \leq d(v,c_j') + d(c_j', \ell_{j})$. Since $x_{v, c_j'} \neq 0$, the LP constraint~\eqref{eq:tooclose} ensures $d(v, c_j') \geq d(c_j', \ell_j)$. Thus $d(v, \ell_j) \leq 2d(v,c_j') \leq 2s_j$.

Observe that~\Cref{alg:rounding} ensures that either $\ell_{j}$ is open, or some center $c_{j} \in C_{j}$ is open. Also $d(\ell_j,c_j) \leq d(\ell_j,c_j')(1 + \epsilon) \leq d(v, c_j')(1 + \epsilon) \leq s_j(1 + \epsilon)$, where the first inequality follows by the definition of $C_j$, and the second inequality again follows from the LP constraint~\eqref{eq:tooclose}. Together, we obtain $d(v, c_j) \leq d(v,\ell_j) + d(\ell_j, c_j) \leq 2s_j + s_j(1 + \epsilon) = 3s_{j}(1 + \epsilon)$. Since this holds for any $j$, and in particular when $j = 1$, it is clear from the description of~\Cref{alg:asgn} that $v$ will be always assigned an open center at distance at most $3s_{1}(1 + \epsilon)$ and the result follows.
\end{proof}

Let $LP(v) := \sum_{c \in V} x_{v, c} d(v, c)$ be the connection cost of $v$ in the LP.
For each case of the assignment given by~\Cref{alg:asgn}, let us define $cost(v)$ to be the following upper bound on $d(v, C)$: 

\begin{enumerate}
\item Some center $c \in D$ is open. In this case, the point goes to the open center $c_{i*} \in D \cap C_{i*}$, where $i^*$ is the minimum index for which there is a center open in $D \cap C_{i^*}$. Given that there is a center open in $D \cap C_{i^*}$, observe that~\Cref{alg:rounding} opens each center $c \in D \cap C_{i^*}$ with probability proportional to $y_c$. 
Thus the expected connection cost is exactly $s_{i^*}$. Let $cost(v) = s_{i^*}$ in this case. 

\item No center in $D$ is open, but some leader $\ell_j$ is open, for some $j \in A$. The point $v$ is then assigned to the leader $\ell_{j^*}$, where $j^* \in A$ is the minimum index $j^*$ so that $\ell_{j^*}$ is open.  Using ~\Cref{lemma:dist}, the connection cost for $v$ is upper bounded by $2s_{j^*}$. Let $cost(v) = 2s_{j^*}$ in this case. 

\item Any other scenario. Again, using~\Cref{lemma:dist}, we upper bound the cost by $3s_{1}(1 + \epsilon)$. Let $cost(v) = 3(1+\eps)s_{1}$ in this case. 

\end{enumerate}

The rest of the section will be devoted to upper bounding the ratio $\E[cost(v)] / LP(v)$, which will yield \Cref{thm:algo}. 

\begin{theorem}
For every $v$, the ratio $(\E[cost(v)]/ LP(v)) \leq \APXs(1 + O(\eps)) \approx \APX(1 + O(\eps))$. 
\label{thm:lp-cost}
\end{theorem}

As in~\Cref{alg:asgn}, let $A$ be the set of indices $i \in [\ell]$ with $s_i \leq 1.5s_{1}$, and let $B$ be the set of indices $i \in [\ell]$ with $s_i > 1.5s_1$.
Fix the ``degrees'' $d_A = |A|$, $d_B = |B|$, the ``flows'' $\mu_A = \sum_{i \in A} \mu_i$, $\mu_B = \sum_{i \in B} \mu_i$. 
By our previous assumption, since $s_1 \leq s_2 \leq s_3 \ldots \leq s_{\ell}$, it follows that $A = [|d_A|]$ and $B = [\ell \setminus |d_A|]$.

Both $\E[cost(v)]$ and $LP(v)$ will be entirely determined by $p, d_A, d_B, (\mu_i)_{i \in [\ell]}$, and $(s_i)_{i \in [\ell]}$. 
So our main question is: for a given $p$ (that our algorithm can control), 
which values for $d_A, d_B, (\mu_i)_{i \in [\ell]}$, and $(s_i)_{i \in [\ell]}$ will maximize the ratio $\E[cost(v)]/ LP(v)$? 

It is simple to see that $LP(v) = \sum_{i \in [\ell]} \mu_i s_i$, 
while $\E[cost(v)]$ is a complex function of $d_A, d_B, (\mu_i)_{i \in [\ell]}$, and $(s_i)_{i \in [\ell]}$. 
We will show that the values for these variables maximizing the ratio $\E[cost(v)] / LP(v)$ must exhibit a certain structure, which finally leads to the desired upper bound. 
This will be achieved via a series of factor-revealing programs that maximize the ratio over a subset of variables while the others are considered fixed. The analysis is spread across~\Cref{subsec:s}, \Cref{subsec:mu} and \Cref{subsec:final}.

\subsection{Simplifying distances}
\label{subsec:s}
We will first consider such a program where $p, \mu_1, \dots, \mu_{\ell}, d_A, d_B$ are fixed and the variables are $(s_i)_{i \in [\ell]}$.
Henceforth, we assume without loss of generality that $s_i \leq 3s_1(1 + \epsilon)$ for each $i \in [\ell]$. If this is not satisfied for some $i \in [\ell]$, then for every $c \in C_i$ with $x_{v,c} \neq 0$ such that $d(v,c) \geq 3s_1(1 + \epsilon)$, we set $d(v,c) = 3s_1(1 + \epsilon)$. First, observe that this operation can only decrease $LP(v)$. Also, notice that by the description of~\Cref{alg:asgn}, this change can only increase $\E[cost(v)]$. To see this, fix any set of centers opened by~\Cref{alg:rounding}. For this fixed set of open centers, $cost(v)$ changes only if some center of $C_i$ is open, in which case~\Cref{alg:asgn} may assign $v$ to a center of $C_i$, which may increase $cost(v)$. It follows that the ratio $\E[cost(v)]/LP(v)$ does not decrease after this operation.\\
The following lemma shows that $\E[cost(v)]$ is a linear function of $(s_i)_{i \in [\ell]}$.

\begin{restatable}{lemma}{distancelemma}
$\E[cost(v)] = \sum_{i=1}^{\ell} \gamma_i s_i$ where each $\gamma_i$ is a function of $p,d_A,d_B,(\mu_i)_{i \in [\ell]}$ only. 

\label{lem:linear_s}
\end{restatable}

\begin{proof}
Consider some $i \in [\ell]$. The probability that~\Cref{alg:rounding} decides to open $\ell_i$ is exactly $p$. With remaining probability $1-p$, exactly one center of $C_i$ is opened, and each center $c \in C_i$ is chosen with probability $y_c = x_{v, c}$. Thus the probability that some $c \in D \cap C_i$ is open is exactly $(1-p)\sum_{c \in (D \cap C_i)}{y_c} = (1-p)\sum_{c \in (D \cap C_i)}{x_{v, c}} = (1-p)\sum_{c \in C_i} x_{v, c} = (1-p)\mu_i$.

If a center of $D$ is open,~\Cref{alg:asgn} assigns $v$ to an open center in $D \cap C_{i^*}$ where $i^*$ is the first index $i^*$ for which there is an open center in $D \cap C_{i^*}$. For each $i \in [\ell]$, $i = i^*$ if and only if for every $1 \leq i' < i$,~\Cref{alg:rounding} does not choose a center in $D \cap C_{i'}$, and~\Cref{alg:rounding} chooses some center in $D \cap C_i$. Since the choices made by~\Cref{alg:rounding} are independent across different $i' \in [\ell]$, it follows this probability is $z_i = (\prod_{1 \leq i' < i} (1 - (1-p)\mu_{i'}))(1-p)\mu_i$, and hence depends only on $p, \mu_i, i \in [\ell]$. In this case, recall that we upper bound the cost by $s_i$.

When no center of $D$ is open,~\Cref{alg:asgn} assigns $v$ to $\ell_{j*}$ where $j^*$ is the first index $j \in A$ for which $\ell_j$ is open. For any fixed $j \in [\ell]$, we claim that the probability $\lambda_j$ that no center of $D$ is open and $j = j^*$ depends only on $p, \mu_i, i \in [\ell]$ and $d_A$. To see this, first note that if $j \in B$, then $j \neq j^*$, since the algorithm never assigns points to $l_j$, and hence $\lambda_j = 0$. If $j \in A$, the event that $j = j^*$ and no center of $D$ is open happens only when for every $1 \leq j' < j$,~\Cref{alg:rounding} does not choose $\ell_{j'}$ or a center of $D \cap C_{j'}$, ~\Cref{alg:rounding} chooses $\ell_j$, and for $j < j'\leq \ell$ no center in $D \cap C_{j'}$ is opened. This probability is exactly $\lambda_j = (\prod_{1 \leq j' < j} (1-p)(1-\mu_{j'})) \cdot p \cdot (\prod_{j < j'\leq \ell} (1 - (1-p)\mu_{j'})$. In this case, recall that we upper bound the cost by $2s_j$.

Finally, if no center of $D$ is open and no leader $\ell_j$ is open for $j \in A$,~\Cref{alg:asgn} assigns $c$ to an open center in $C_1$. This happens with probability $\gamma = \prod_{j \in A} (1-p)(1 - \mu_j) \prod_{j \in B} (1 - (1-p) \mu_j)$, and the cost is upper bounded by $3(1 + \eps)s_1$.

Now observe that for every $j \geq 2$, the coefficient $\gamma_j$ of $s_j$ is exactly $z_j + 2\lambda_j$, and $\gamma_1 = z_1 + 2\lambda_1 + 3(1+\eps)\gamma$. 

This allows us to conclude that each $\gamma_i, i \in [\ell]$ is a function of $p,d_A, d_B$ and $(\mu_i)_{i\in [\ell]}$ alone.
\end{proof}

Recall that we assumed that $s_1 \leq s_2 \leq s_3 \leq \ldots \leq s_{\ell = d_A + d_B}$. 
Then finding the worst case ratio for the point $v$ over $(s_i)_{i \in [\ell]}$ can be cast as the following optimization problem whose variables are the distances $(s_i)_{i \in [\ell]}$, which we refer to as $O$:

\begin{center}
\noindent\fbox{

  \begin{minipage}{0.95\textwidth}
  \begin{align*}
  \max \quad & \frac{\sum_{i = 1}^{\ell} \gamma_i s_i}{\sum_{i = 1}^{\ell} \mu_{i} s_{i}} \\ 
  \mbox{s.t.} \quad & s_{1} \leq s_{2} \leq s_{3} \leq  \ldots \leq  s_{{d_A}} \leq 1.5s_{1} \\
  & 1.5s_{1} \leq s_{d_A + 1} \leq s_{{d_A + 2}} \leq  \ldots \leq s_{{d_A + d_B}} \leq 3s_{1}(1 + \epsilon) \\
  & s_1 > 0
  \end{align*}
  
  \end{minipage}
       }
\end{center}
(Note that we allow $s_i$ for $i \in B$ to take the value exactly $1.5 s_1$, which is not allowed in the definition of $B$. However, it only increases the optimal value.)

It is now clear that this optimization problem $O$ (where we fix $p,d_A,d_B, ( \mu_i )_{i \in [\ell]}$) is a linear program with the variables $(s_i)_{i \in [\ell]}$, for we can simply add the additional constraint $\sum_{i = 1}^{\ell} \mu_{i} s_{i} = 1$ and maximize the numerator of the objective. Since any linear program on $\ell$ variables has an optimal solution that is an extreme point where exactly $\ell$ linearly independent constraints are tight, this allows us to greatly simplify the structure of optimal solutions.

\begin{restatable}{lemma}{lemmaextremepoint}
There is an optimal solution to $O$ which satisfies $s_{i} \in \{s_1, 1.5s_1, 3s_1\}$ for each $i \in [\ell]$.
\end{restatable}

\begin{proof}
Consider the following linear program which is equivalent to $O$.
\begin{center}
\noindent\fbox{

  \begin{minipage}{0.95\textwidth}
  \begin{align}
  \max \quad & {\sum_{i = 1}^{\ell} \gamma_{i} s_{i}} \nonumber \\
  \mbox{s.t.} \quad & {\sum_{i = 1}^{\ell} \mu_{i} s_{i}} = 1 \label{eq:extreme_zero} \\
  & s_{1} \leq s_{2} \leq s_{3} \leq \ldots \leq s_{{d_A}} \leq 1.5s_{1} \label{eq:extreme_one} \\
  & 1.5s_{1} \leq s_{d_A + 1} \leq s_{d_A + 2} \leq \ldots \leq s_{d_A + d_B} \leq 3s_{1}(1 + \epsilon) \label{eq:extreme_two} \\
  & s_1 \geq 0  \label{eq:extreme_three}
  \end{align}
  \end{minipage}
       }
\end{center}
Let $(s_i^*)_{i \in [\ell]}$ be an \emph{extreme point optimal solution} for this linear program. We claim that $s_i^* \in \{s_1, 1.5s_1, 3s_1(1 + \epsilon)\}$ for each $i \in [\ell]$.
We show this as follows. First, there are a total of $\ell = d_A + d_B$ variables in the linear program. An extreme point, therefore, must be the intersection of $\ell$ linearly independent constraints. There are a total of $\ell + 1$ constraints captured by~\eqref{eq:extreme_one} and~\eqref{eq:extreme_two}, and two more constraints~\eqref{eq:extreme_zero} and~\eqref{eq:extreme_three}. At an extreme point, at most $3$ of these constraints may not be tight. The constraint $s_1 > 0$ cannot be tight, since it forces the objective to be zero, which is clearly not optimal. Also, every constraint in the family constraints~\eqref{eq:extreme_one} cannot be tight, similarly, every constraint in the family of constraints~\eqref{eq:extreme_two} cannot be tight, since that would imply $s_1 = 0$. It follows that exactly one constraint among the constraints~\eqref{eq:extreme_one} is not tight, and exactly one constraint among the constraints~\eqref{eq:extreme_two} is not tight. But this means that $s_i^* \in \{s_1, 1.5s_1, 3s_1(1 + \epsilon)\}$ for each $i \in [\ell]$, which proves the result.
\end{proof}

\subsection{Simplifying LP values and degrees}
\label{subsec:mu}

Once we know that the worst-case ratio must involve distances from $\{s_1, 1.5s_1, 3s_1(1 + \epsilon)\}$, we proceed as follows. 

Without loss of generality, assume that $s_1 = 1$. Let $A',B'$ and $C'$ denote the set of indices $i \in [\ell]$ which have $s_i = 1, 1.5$ and $3(1 + \epsilon)$ respectively. We can now compute the worst case upper bound $\E[cost(v)]$ for the point $v$ in this special case as follows. For some distance parameter $\rho$, let $t_{\rho} := \Pr[cost(v) \geq \rho]$.
We say that $v$ has a \emph{direct connection} in a set $I' \subseteq [\ell]$ if there exists a center $c \in C_i \cap D$ for some $i \in I'$ which is open. Further we say that a leader of $I'$ is open if there exists an $i \in I'$ such that the leader $\ell_i$ is open. With these definitions, we have the following observations.

\begin{lemma} The following equations hold. 
\vspace{-5pt}
\begin{align*}
t_{1} =& 1 \\
t_{1.5} =& \Pr[\text{No direct connection in $A'$}]\\
t_{2} =& \Pr[\text{No direct connection in $A' \cup B'$}]\\
t_{3}  =& \Pr[\text{No direct connection in $A' \cup B'$, but direct connection in $C'$}] +\\ &\Pr[\text{No direct connection in $A' \cup B' \cup C'$, and no leader of $A'$ is open}].
\end{align*}
\end{lemma}
\begin{proof}
The point $v$ always pays a cost at least $1$, so the ``baseline'' cost is $1$, which is paid with probability $1$. If there is no direct connection to any center in $A'$, then the point must clearly pay at least $1.5$ (possibly using a direct connection to an open center of $B'$). If there is no direct connection to centers in $A' \cup B'$, the point must pay at least $2$ (possibly using a connection to an open leader of $A$). Finally, the point can pay $3$ or $3(1 + \epsilon)$ in two different ways. The first scenario is when there is no direct connection to centers in $A' \cup B'$, but a direct connection to a center of $C'$. Recall that~\Cref{alg:asgn} always assigns a point to a center whenever a direct connection is open, therefore in this case the point gets assigned to some center of $C'$ at  distance $3(1 + \epsilon)$. The second scenario is when there is no direct connection in $A' \cup B'\cup C'$, and no leader of $A'$ is open. In this case the point gets connected at distance $\geq 3$ (for instance it may be assigned to a leader of $B'$).
\end{proof}
Let $cost'(v) = t_{1} + 0.5t_{1.5} + 0.5 t_{2} + t_{3}$. Then it is clear that $cost'(v)(1 + \epsilon) \geq \E[cost(v)]$, because $cost(v)$ is always between $1$ and $3(1+\eps)$, and $cost'(v) = \int_{x=0}^3 \Pr[cost(v) \geq x] \,dx$. 
It remains to analyze and upper bound these probabilities. Let 
\begin{align*}
p_{A'} &= \text{Pr[No direct connection in $A'$]} \\
p_{B'} &= \text{Pr[No direct connection in $B'$]} \\
p_{C'} &= \text{Pr[No direct connection in $C'$]} \\
q_{A'} &= \text{Pr[No direct connection in $A'\; \land $ No leader open in $A'$]}.
\end{align*}
Then by the above discussion and definitions we have $t_{1.5} = p_{A'}$, $t_{2} = p_{A'}p_{B'}$, and $t_{3} = (1-p_{C'}) p_{A'} p_{B'} + p_{C'} q_{A'} p_{B'}$, which yields
\[
cost'(v) = 1 + 0.5 p_{A'}  + 0.5 p_{A'} p_{B'}  + ((1 - p_{C'}) p_{A'} p_{B'} + p_{C'} q_{A'} p_{B'}).
\]

For fixed $d_{A'}, \mu_{A'}, \mu_{B'}, \mu_{C'}$ (which decides $LP(v) = \mu_{A'}  + 1.5 \mu_{B'}  + 3(1+\eps)\mu_{C'}$), we seek the worst-case configurations over $(\mu_i)_{i \in [\ell]}$, $d_{B'}$ and $d_{C'}$ that maximize $cost'(v)$. 
Note that $p_{A'}$ (and $q_{A'}$), $p_{B'}$, $p_{C'}$ depend only on the $\mu_i$ values of $A', B', C'$ respectively, so we can treat them separately. We can deduce the following conclusions about the worst-case configurations. 
Note that for any $I' \subseteq [\ell]$, $p_{I'} = \text{Pr[no direct connection in $I'$]} = \prod_{i \in I'} (1 - (1-p)\mu_i)$.

\begin{itemize}
\item For each of $p_{A'}$, $q_{A'}$, and $p_{B'}$, the coefficient (treating other terms as constants) in $cost'(v)$ is nonnegative, so the worst case happens when they are maximized. 

\item Let us first compute an upper bound on $p_{A'} = \prod_{i \in A'} (1 - (1-p)\mu_i)$. Fix $d_{A'} = |A'|$ and $\mu_{A'} = \sum_{i \in A'} \mu_i$. By a simple application of the AM-GM inequality, it follows that $p_{A'}$ is maximized when $\mu_i = \nf{\mu_{A'}}{d_{A'}}$ for each $i \in A'$, and this value equals $(1 - \nf{(1-p)\mu_{A'}}{d_{A'}})^{d_{A'}}$. Defining $p_{B'}$, $d_{B'}$ similarly, we derive similar expressions for upper bounds on $p_{B'}$.
For $q_{A'}$, note that the event corresponding to $q_{A'}$ happens when there is no direct connection in $A'$ and no leader open in $A'$: the probability of this happening is exactly $\prod_{i \in A'} (1-p)(1-\mu_i)$. Again using the AM-GM inequality, we can upper bound this by $(1-p)^{d_{A'}}(1 - \nf{\mu_{A'}}{d_{A'}})^{d_{A'}}$.

\item For fixed $\mu_{B'}$, for the worst case, since $p_{B'} = (1 - \nf{(1-p)\mu_{B'}}{d_{B'}})^{d_{B'}}$ is increasing with $d_{B'}$, we can let $d_{B'} \rightarrow \infty$, which results in $e^{-(1-p) \mu_{B'}}$. (We cannot do this for $A'$ because of the factor $(1-p)^{d_{A'}}$ in $q_{A'}$.)

\item For $p_{C'}$, notice that the term involving $p_{C'}$ is $p_{C'}(q_{A'}p_{B'} - p_{A'}p_{B'})$, and $q_{A'} \leq p_{A'}$. Thus in order to maximize $cost'(v)$, it is beneficial to reduce $p_{C'}$ as much as possible (given $\mu_{C'}$), and hence we can set $d_{C'} = 1$, which makes $p_{C'}$ the smallest possible value $1 - (1-p)\mu_{C'}$. 
\end{itemize}

Accounting for these changes, we obtain 
\begin{align*}
t_{1.5} & \leq t'_{1.5} =  \bigg(1 - (1-p)\frac{\mu_{A'}}{d_{A'}}\bigg)^{d_{A'}} \\
t_{2} & \leq t'_{2} = \bigg(1 - (1-p)\frac{\mu_{A'}}{d_{A'}}\bigg)^{d_{A'}} e^{(p-1)\mu_{B'}} \\
t_{3} & \leq t'_{3} =  \bigg(1 - (1-p)\frac{\mu_{A'}}{d_{A'}}\bigg)^{d_{A'}} e^{(p-1)\mu_{B'}} \mu_{C'}(1-p) + \left ((1-p)(1 - \frac{\mu_{A'}}{d_{A'}})\right )^{d_{A'}} e^{(p-1)\mu_{B'}}(1 - \mu_{C'}(1-p)).
\end{align*}

Thus, we now only need to compute the min-max value $ \min_p \max_{d, \mu_{A'}, \mu_{B'}, \mu_{C'}}(\nf{cost'(v)}{LP(v)})$ where $cost'(v) = 1 + 0.5 t'_{1.5} + 0.5 t'_{2} + t'_{3}$, and the variables are $p, d_{A'}, \mu_{A'}, \mu_{B'}, \mu_{C'}$ with the constraint $\mu_{A'} + \mu_{B'} + \mu_{C'} = 1$, and $d_{A'}$ being a positive integer. 

\subsection{Final simplifications}
\label{subsec:final}

In the next lemma, we show that in fact we can assume without loss of generality that $\mu_{A'} = 1$ and $\mu_{B'} = \mu_{C'} = 0$, so that this is equivalent to minimizing  over $p$ and maximizing over $d$ the expression $g(p,d) = 1 + (1 - \nf{(1-p)}{d})^d + ((1-p)(1-1/d))^d$. This concludes the proof of \Cref{thm:algo}.

\begin{lemma}\label{lemma:optimize}
 $\min_{p}\max_{d_{A'},\mu_{A'},\mu_{B'},\mu_{C'}}  (\nf{cost'(v)}{LP(v)}) = \min_p \max_{d} g(p,d) \leq \APXs \approx \APX$, where $\APXs := g(p^*, d^*)$ with
$p^* := \nf{(10 - 6\sqrt{2}) }{7}\approx 0.22$ and $d^* := 3$.
\end{lemma}
\begin{proof}
For the sake of convenience and better readability, we denote $A',B',C'$ by $A,B,C$. We begin by upper bounding each term. First, observe that the function $\psi(x) = (1 - ax)^{b}$ is convex when $1 - ax > 0$ and $b \geq 1$. This can be checked easily by computing the second derivative $\psi''(x) = a^2b(b-1)(1-ax)^{b-2} > 0$. Using Jensen's inequality on the interval $[0,1]$, this means that $(1-ax)^b \leq x((1-a)^b - 1) + 1$ whenever $x \in [0,1]$. This means one can upper bound $(1 - (1-p)\frac{\mu_A}{d_A})^{d_A}$ by $\mu_A((1-\frac{1-p}{d_A})^{d_A} - 1) + 1$. Similarly, the function $e^{ax}$ is convex for any $a \in \mathbb{R}$, and by Jensen's inequality we have $e^{ax} \leq (e^a - 1)x + 1$. This gives $e^{(p-1)\mu_B} \leq (e^{p-1} - 1)\mu_B + 1$. Finally, we note similarly that $(1 - \frac{\mu_A}{d_A})^{d_A} \leq \mu_A((1 - \frac{1}{d_A})^{d_A} - 1) + 1$. Plugging these into the expression for $cost'(v)$, we obtain the following upper bound for $\nf{cost'(v)}{LP(v)}$.

\begin{align*}
\frac{cost'(v)}{LP(v)}  \leq \psi(p,d_A,\mu_A, \mu_B, \mu_C) &= \frac{1}{\mu_A + 1.5\mu_B + 3\mu_C} \Bigl[1 +  0.5\bigl(1-\mu_A+\mu_A \left(1-\frac{1-p}{d_A}\right)^{d_A}\bigr) \\ &+ 0.5\bigl(1-\mu_B+e^{p-1} \mu_B\bigr) \bigl(1-\mu_A+\mu_A \left(1-\frac{1-p}{d_A}\right)^{d_A}\bigr) \\
  &+\mu_C(1-p)\bigl(1-\mu_B+e^{p-1} \mu_B\bigr) \bigl(1-\mu_A+\mu_A \left(1-\frac{1-p}{d_A}\right)^{d_A}\bigr) \\&+ (1 - \mu_C(1-p))\bigl(1-\mu_B+e^{p-1} \mu_B\bigr) \bigl(\left(1-p\right)^{d_A}-\mu_A \left(1-p\right)^{d_A}\\&+\left(1-\frac{1}{d_A}\right)^{d_A} \mu_A \left(1-p\right)^{d_A}\bigr)\Bigr] \\
  &=: \frac{\zeta(p,d_A,\mu_A,\mu_B,\mu_C)}{\mu_A + 1.5\mu_B + 3\mu_C}
\end{align*}

Observe that $g(p,d) = 1 + (1 - \frac{1-p}{d})^d + ((1-p)(1-\frac{1}{d}))^d$ is obtained by setting $d_A = d$, $\mu_A = 1$, $\mu_B = \mu_C = 0$ in $\psi(p,d_A,\mu_A,\mu_B,\mu_C)$. 
    % \enote{Somehow it is recommended to avoid starting a sentence with a math expression.}
Then $g(p,d)$ is a function of two variables: we prove that the min-max value $\min_p \max_d g(p,d) = \APXs$ is 
where $\APXs := g(p^*, d^*) \approx \APX$.

\begin{lemma}\label{lemma:minmax}
$\min_p \max_d g(p,d) = \APXs$, and the min-max value is attained when $p = p^* =  \frac{1}{7}(10 - 6\sqrt{2})$ and $d = d^* = 3$.
\end{lemma}

\begin{proof}
Let $d^* = 3$, $p^* =  \frac{1}{7}(10 - 6\sqrt{2})$, and let 
$\alpha := g(p^*, d^*) \approx \APX$.
When $d = 3$, $g(p,d)$ is a cubic function of $p$ which is easily shown to have a minimum value at $p = p^*$, showing that $\max_{d} \min_p g(p, d) \geq \alpha$. 

Similarly, $g(p^*,d)$ is a function of $d$ alone. Using the inequality $1-x \leq e^{-x}$, observe that $1 + e^{p^*-1} + \frac{1}{e}(1-p^*)^d$ is  an upper bound on $g(p^*,d)$. When $d \geq 7$, this quantity is at most $1.542$. Manually checking values of $g(p^*, d)$ for $d \leq 7$, the maximum is attained when $d = 3$, which implies that $\min_p \max_d g(p, d) \leq \alpha$. Since the min-max value is at least the max-min in general and we proved 
$\max_{d} \min_p g(p, d) \geq \alpha \geq \min_p \max_d g(p, d)$, both values are exactly $\alpha$. 
\end{proof}

Next we will show that for $p = p^* =  \frac{1}{7}(10 - 6\sqrt{2}) $ and any choice of $d_A, \mu_A, \mu_B, \mu_C$ with $\mu_A + \mu_B + \mu_C = 1$ , we have $\psi(p,d_A, \mu_A, \mu_B, \mu_C) \leq \APXs$. Henceforth we fix $p = p^*$ and try to maximize $\psi(p,d_A,\mu_A,\mu_B,\mu_C)$.

First, we show that there is a maximizer of $\psi(.)$ with $\mu_C = 0$. We need the following lemma for the proof.

\begin{lemma}
The functions $\zeta(p,d_A, \mu_A, 1 - \mu_A - \mu_C, \mu_C)$ and $\zeta(p,d_A, 1 - \mu_B - \mu_C, \mu_B, \mu_C)$ are both convex with respect to $\mu_C$.
\end{lemma}

\begin{proof}
We show convexity by computing second derivatives. Consider $\frac{\partial^2 \zeta(p,d_A, \mu_A, 1 - \mu_A - \mu_C, \mu_C)}{\partial \mu_C^2}$. Since $\zeta(p,d_A, \mu_A, 1 - \mu_A - \mu_C, \mu_C)$ is a quadratic function of $\mu_C$, the second derivative is simply twice the coefficient of $\mu_C^2$. Simple computation shows that this is equal to $$2\left(1-e^{p-1}\right)(1-p)\left(1 - (1-p)^{d_A} + \mu_A\left(\left(1-\frac{1-p}{d_A}\right)^{d_A}- 1 - \left(1-\frac{1}{d_A}\right)^{d_A}(1-p)^{d_A} +  (1-p)^{d_A}\right)\right)$$. Let $a = 1 - p$ and $b = \frac{1}{d_A}$. Then it is sufficient to show that $ \rho(a,b,d_A) = 1 - a^{d_A} + \mu_A\left((1 - ab)^{d_A} - 1 - a^{d_A}(1-b)^{d_A} + a^{d_A}\right)$ is positive. First note that $(1-ab)^{d_A} \geq (1-b)^{d_A}$. Plugging this and simplyfing yields the expression $$1 - a^{d_A} - \mu_A(1 - (1-b)^{d_A})(1 -a^{d_A}) = (1 - a^{d_A})(1 - \mu_A(1 - (1-b)^{d_A}) > 0$$ for our choices and suitable ranges of of $a,b, d_A$. 

Similarly, consider the partial derivative  $\frac{\partial^2 \zeta(p,d_A, 1 - \mu_B - \mu_C, \mu_B , \mu_C)}{\partial \mu_C^2}$. Again a simple computation yields that this is equal to
$$2(1 - (1-e^{p-1})\mu_B)(1-p)\left(1 - (1-\frac{1-p}{d_A})^{d_A}  - (1-p)^{d_A} + (1-\frac{1}{d_A})^{d_A}(1-p)^{d_A}\right)$$.

The last term is  equal to $ \rho(a,b,d_A) = 1 - (1 - ab)^{d_A} - a^{d_A} + a^{d_A}(1 - b)^{d_A}$

Observe that $\rho(0,b,d_A) = 0$, and $\rho(1,b,d_A) = 0$. We now show $\rho$ is concave with respect to $a$ whenever $a \in (0,1)$. Consider $\frac{\partial^2 \rho(a,b,d_A)}{\partial a^2}$. This is equal to $-b^2d_A(d_A - 1) (1 - ab)^{d_A - 2} -  d_A(d_A - 1)a^{d_A - 2} (1 - (1-b)^{d_A})$. Both of these terms are negative when $d_A \geq 1$, and hence $\rho$ is indeed concave with respect to $a$. It follows that $\rho(a,b,d_A) > 0$ for any $a \in (0,1)$, and hence $\zeta(p,d_A,1 - \mu_B - \mu_C, \mu_B, \mu_C)$ is convex with respect to $\mu_C$.
\end{proof}

\begin{observation}
The functions $\zeta(p,d_A, \mu_A, 1 - \mu_A - \mu_C, \mu_C) - \lambda(\mu_A + 1.5(1 - \mu_A - \mu_C) + 3\mu_C)$ and $\zeta(p,d_A, 1 - \mu_B - \mu_C, \mu_B, \mu_C) - \lambda(1 - \mu_B - \mu_C + 1.5 \mu_B + 3\mu_C)$ are convex with respect to $\mu_C$ for any fixed $\lambda \in \mathbb{R}$.

\end{observation}

\begin{proof}
Follows from the fact that adding a linear function to a convex function results in a convex function.
\end{proof}

\begin{lemma}
There is a maximizer for $\psi(p,d_A, \mu_A,\mu_B, \mu_C)$ where $\mu_C = 0$.
\end{lemma}

\begin{proof}
Let $\lambda^*$ be the maximum value of $\psi(p,d_A, \mu_A,\mu_B, \mu_C)$. Then we must have $\zeta(p,d_A, \mu_A, 1 - \mu_A - \mu_C, \mu_C) - \lambda^*(\mu_A + 1.5(1 - \mu_A - \mu_C) + 3\mu_C) \geq 0$ and $\zeta(p,d_A, 1 - \mu_B - \mu_C, \mu_B, \mu_C) - \lambda^*(1 - \mu_B - \mu_C + 1.5 \mu_B + 3\mu_C) \geq 0$ for some choice of $d_A, \mu_A, \mu_B, \mu_C$. Since  $\zeta(p,d_A, \mu_A, 1 - \mu_A - \mu_C, \mu_C) - \lambda^*(\mu_A + 1.5(1-\mu_A-\mu_C) + 3\mu_C)$ is convex with respect to $\mu_C$, it follows that there is another point $d_A, \mu_A', \mu_B',\mu_C'$ which attains a higher value of  $\zeta(p,d_A, \mu_A', \mu_B', \mu_C') - \lambda^*(\mu_A' + 1.5 \mu_B' + 3\mu_C' )$ and satisfies $\mu_C' = 0$ or $\mu_C' = 1 - \mu_A'$. Using the fact that $\zeta(p,d_A, 1 - \mu_B - \mu_C, \mu_B, \mu_C) - \lambda^*(1 - \mu_B - \mu_C + 1.5 \mu_B + 3\mu_C)$ is convex with respect to $\mu_C$, we can further obtain a triplet $(\mu_A'',\mu_B'',\mu_C'')$ that still satisfies both inequalities and further $\mu_C'' = 0$ or both $\mu_C'' = 1 - \mu_B''$ and $\mu_C'' = 1 - \mu_A''$. These two relationships imply that either $\mu_C'' = 0$, or both $\mu_A'' = \mu_B'' = 0$, so that $\mu_C'' = 1$. However, the latter is impossible, since $\zeta(p,d_A, 0,0,1) = \frac{1}{3}(3 - p + p(1-p)^d) < 1$ which is clearly less than $\lambda^*$. Hence $\mu_C'' = 0$.
\end{proof}

Henceforth, we concern ourself with $\psi(p,d_A, \mu_A, 1 - \mu_A,0)$. For convenience we will use the notation $\psi(p,d,\mu)$ instead, where $d:= d_A$ and $\mu:= \mu_A$.

\begin{lemma}
When $d \geq 3$ and $p = p^*$, $\frac{\partial\psi(p,d,\mu)}{\partial \mu} > 0$ for all $\mu \in (0,1)$.

\end{lemma}

\begin{proof}

Again, define $a = (1-p)$ and $b = \frac{1}{d}$ and additionally define $c = e^{p-1}$. First, we evaluate $\psi(p,d,\mu)$ by plugging $\mu_C = 0$, $\mu_B = 1- \mu$, $\mu_A = \mu$, $d_A = d$ in $\psi(p,d_A,\mu_A,\mu_B,\mu_C)$. We obtain
\begin{align*}
\psi(p,d,\mu) &= \frac{1}{1.5 - 0.5\mu}(1.5 + 0.5c + ca^d + \mu((1-ab)^d(0.5 + 0.5c) - c + ca^d(1-b)^d + a^d(1-2c)) + \\
&\frac{\mu^2}{2} (1-c) ((1-ab)^d - 1 + 2a^d((1-b)^d -  1)))
\end{align*}
Now we take the partial derivative $\frac{\partial \psi(p,d,\mu)}{\partial \mu}$. First, let us write $\psi(p,d,\mu) := \frac{\psi_1(p,d,\mu)}{1.5 - 0.5\mu}$. We have
$$\frac{\partial \psi_1(p,d,\mu)}{\partial \mu} = (1-ab)^d(0.5 + 0.5c) - c + ca^d(1-b)^d + a^d(1-2c) + \mu (1-c) ((1-ab)^d - 1 + 2a^d((1-b)^d -  1))$$

Notice that $\frac{\partial \psi(p,d,\mu)}{\partial \mu}$ = $\frac{1}{(1.5 - 0.5\mu)^2}((1.5 - 0.5 \mu) \frac{\partial \psi_1(p,d,\mu)}{\partial \mu} + 0.5 \psi_1(p,d,\mu)) := \frac{\psi_2(p,d,\mu)}{(1.5 - 0.5 \mu)^2} $.

Here $\psi_2$ is a quadratic function of $\mu$ of the form $\lambda_1 \mu^2 + \lambda_2\mu + \lambda_3$, where
\begin{align*}
\lambda_1 &= -\frac{1}{4}(1-c) ((1-ab)^d - 1 + 2a^d((1-b)^d -  1))\\
\lambda_2 &= 1.5(1-c) ((1-ab)^d - 1 + 2a^d((1-b)^d -  1))\\
\lambda_3 &= 1.5 ((1-ab)^d(0.5 + 0.5c) - c + ca^d(1-b)^d + a^d(1-2c)) + 0.5(1.5 + 0.5c + ca^d)
\end{align*}

Observe that $\lambda_1 > 0$ and $\frac{\lambda_2}{-2\lambda_1} = 3$. This means that $\psi_2$ attains it minimum at $\mu = 3$ and is decreasing with respect to $\mu$ when $\mu \in (0,1)$. We compute $\psi_2(p,d,1)$.

\begin{align*}
\psi_2(p,d,1) = \lambda_1 + \lambda_2 + \lambda_3
&= (2 - 0.5c)(1-ab)^d + (2.5 - c)a^d(1-b)^d-a^d - \frac{1}{2}
\end{align*}

Recall that $b = \frac{1}{d}$, and $0 < a,b,c < 1$.  The function $(1-k/x)^x$ is increasing when $x \geq 1$, for any $k > 0$. Since $d \geq 3$, we obtain the following lower bound on $\psi_2(p,d,1)$.

$$\psi_2(p,d,1) \geq \psi_2'(p,d,1) = (2 - 0.5c)(1-a/3)^3 + (2.5-c)a^d(1-1/3)^3 - a^d - \frac{1}{2}$$

Plugging in $p = p^*$ in $\psi_2'(p,d,1)$, this expression evaluates approximately to $0.214 - 0.395(1 - p^*)^d$. Since $d \geq 3$, this quantity is at least $\psi_2'(p^*,3,1) \approx 0.025 > 0$. Since $\psi_2(p^*,d,\mu)$ decreases with $\mu$ and $\psi_2(p^*,d,1) > 0$ whenever $d \geq 3$, it follows that $\frac{\partial \psi(p^*,d,\mu)}{\partial \mu} > 0$ for every $\mu \in (0,1)$ and $d \geq 3$. 
\end{proof}

\begin{lemma}
When $d \geq 3$ and $p = p^*$, $\psi(p,d,\mu) \leq \APXs$ for all $\mu \in (0,1)$.
\end{lemma}

\begin{proof}
The proof is straightforward. Since we showed in the previous lemma that $\frac{\partial \psi(p,d,\mu)}{\partial \mu} > 0$, for each $\mu \in (0,1)$ when $p = p^*$ and $d \geq 3$, $\psi(p^*,d,\mu)$ is maximized when $\mu = 1$. When $\mu = 1$ and $p = p^*$, we have $\psi(p,d,1) = g(p^*,d) \leq \APXs$ from~\Cref{lemma:minmax}.
\end{proof}

\begin{lemma}
When $d = 1,2$ and $p = p^*$, $\psi(p,d,\mu) \leq \APXs$ for all $\mu \in (0,1)$.
\end{lemma}

\begin{proof}
The proof is by direct computation. When $d = 1$, and $p = p^*$, we obtain the function $\psi(p^*,1,\mu) \leq \frac{(4.173 - 0.462\mu - 1.277\mu^2)}{3 - \mu}$. Finding the maximum of the right hand side is a single variable maximization problem easily solved by computing first derivatives. We obtain that the maximum is attained when $\mu \approx 0.39$, and the maximum is at most $1.46$.

When $d = 2$ and $p = p^*$, similarly we obtain the function $\frac{4.02 - 0.128\mu - 0.842\mu^2}{3 - \mu}$ as an upper bound, which is maximized when $\mu \approx 0.836$ and the maximum is at most $1.537 < \APXs \approx \APX$. Thus the inequality holds.
\end{proof}

Finally, combining all these results, we have shown that the function $\psi(p,d_A, \mu_A, \mu_B, \mu_C)$ is at most $\APXs \approx \APX$ when $p = p^*$, for any (valid) choice of $d_A, \mu_A, \mu_B, \mu_C$. This completes the proof of~\Cref{lemma:optimize}.
\end{proof}
\section{Hardness}
In this section, we prove the following hardness result for \KM. 

\hardnessmain*

Our result follows from the following hardness result for $k$-Hypergraph Vertex Cover ($k$-HVC), where we allow parallel hyperedges.

\begin{theorem}
Assuming the Unique Games Conjecture, for any $\eps > 0$ and $d \in \Z_{\geq 3}$, given a $d$-uniform hypergraph $H = (V, E)$, it is NP-hard to distinguish the following two cases.
\begin{itemize}
    \item YES: There exists $U \subseteq V$ with $|U| = \frac{|V|}{d - 1}$ that intersects at least a $(1 - \eps)$ fraction of hyperedges. 
    \item NO: For any $\alpha \in [0, 1]$, any subset $U \subseteq V$ with $\alpha |V|$ vertices intersects at most a $1 - (1 - \alpha)^d + \eps$ fraction of hyperedges. 
\end{itemize}
\label{thm:hvc}
\end{theorem}

A standard reduction shows that \Cref{thm:hardness} immediately follows from \Cref{thm:hvc}. 

\begin{proof} [Proof of \Cref{thm:hardness}] 
Given a $d$-uniform hypergraph $H = (V, E)$ from \Cref{thm:hvc}, we first perform the following operation. 
\begin{itemize}
    \item For each hyperedge in $E$, add $M$ copies of the same hyperedge for $M = (|V||E|)^5$. 
    
    \item For each triple $\{ u, v, w \} \subseteq V$ that do not have a corresponding hyperedge, add a hyperedge. 
\end{itemize}
This ensures that (1) $|E| \geq \omega(|V|)$ and (2) any two vertices are contained in the same hyperedge, while maintaining the claimed hardness up to an additive loss of $o(1)$. After such a transformation, our instance for $\KM$ is simply the vertex-hyperedge incidence graph of $H$ with $k = \frac{|V|}{d-1}$. 

In the YES case, by placing $k$ centers at the promised set $U$, at least a $(1 - \eps - o(1))$ fraction of vertices from $E$ have distance $1$ to $U$, and all the other vertices have distance at most $4$ from $U$. Therefore, the average distance to the open centers is $1 + O(\eps)$.

In the NO case, consider any solution $U$ that opens $\frac{1 - p}{d - 1} |V|$ centers in $V$ and 
$\frac{p}{d - 1} |V|$ in $E$ for some $p \in [0, 1]$. Letting $\alpha = \frac{1 - p}{d - 1}$ ensures that at least a $(1 - \frac{1 - p}{d - 1})^d - \eps - o(1)$ fraction of hyperedges have a distance strictly greater than one from $U$. Letting $\alpha = \min(1, \frac{1 - p}{d - 1} + d \frac{p}{d - 1})$ ensures that at least a $(1 - \min(1, \frac{1 - p}{d - 1} + d \frac{p}{d - 1}))^d - \eps - o(1)$ fraction of hyperedges have a distance strictly greater than two from $U$. Since all distances are integers, the average distance to $U$ is at least 
\[h(p,d) = 
1 + \bigg(1 - \frac{1 - p}{d - 1}\bigg)^d + \bigg(1 - \min(1, \frac{1 - p + pd}{d - 1})\bigg)^d - O(\eps).\]

We thus obtain a hardness of $\max_{d \in \mathbb{Z}_{\geq 3}} \min_{p \in [0,1]} h(p,d) \approx \numhard - \OO(\epsilon)$ as desired.

% \enote{Proof of 1.4 here.}
\end{proof}

\subsection{Proof of \Cref{thm:hvc}}
It remains to prove \Cref{thm:hvc}. It follows from standard techniques to prove hardness of approximation of constraint satisfaction or covering problems assuming the Unique Games Conjecture, pioneered by O'donnell et al.~\cite{khot2007optimal}. 
Our proof becomes particularly simple due to the use of {\em pairwise independence}~\cite{austrin2009approximation, Mossel10}. This notion has been crucially used for constraint satisfaction problems, but its use for covering problems (especially with $d > 3$) has been limited. 

\subsubsection{Dictatorship test}
Let $d \geq 3$ be an integer and $R$ be another positive integer. 
Our dictatorship test is a $d$-uniform hypergraph whose set of vertices is $[d - 1]^R$.
Let $\Omega = [d - 1]$. 
A distribution $\mu$ on $\Omega^d$ is called {\em pairwise independent} if $x$ is a random $d$-dimensional vector sampled from $\Omega$, 
for every $i < j \in [d]$ and $t_i, t_j \in [d-1]$, $\Pr[x_i = t_i, x_j = t_j] = \frac{1}{(d-1)^2}$. 

\begin{lemma} [\cite{ghazi2017lp}] For any $d \geq 3$, there exists a pairwise independent distribution $\mu$ on $[d-1]^d$ such that every $x \in [d - 1]^d$ in the support has at least one $i \in [d]$ with $x_i = 1$. 
\end{lemma}

Fix the $\mu$ given by the above lemma, and for another parameter $\delta > 0$ to be determined, let $\mue$ be the $\delta$-noised version of $\mu$ on $\Omega^d$; after sampling $v$ from $\mu$, each coordinate of $x$ is independently resampled with probability $\delta$ to a random value in $\Omega$. 

Let $\mue^{\otimes R}$ be the product distribution on $(\Omega^d)^R$ so that for every $x^1, \dots, x^d \in \Omega^{R}$, $\mue^{\otimes R} (x^1, \dots, x^d) = \prod_{i=1}^R \mu_{\delta}(x^1_i, \dots, x^d_i)$. Then the dictatorship test is a weighted $d$-uniform hypergraph where the weight of a $d$-edge is exactly the probability that it is sampled from $\mue^{\otimes R}$. It can be easily converted to an unweighted hypergraph by duplicating the same hyperedge according to its weight. (Eventually, in the Unique Games reduction, $\eps$ and $R$ will be constants independent of the instance size.)

A \emph{dictator} is to the set $\{ x \in \Omega^R : x_i = 1 \}$ for some $i \in [R]$. Note that every such set intersects at least a $(1 - \delta)$ fraction of hyperedges.

\subsubsection{Fourier analysis preliminaries}
To analyze the soundness of the test, we use the following standard tools from Gaussian bounds for correlated functions from Mossel~\cite{Mossel10}.
We define the correlation between two and more correlated spaces.
\begin{definition}
Given a distribution $\nu$ on $\Omega_1 \times \Omega_2$, the correlation $\rho(\Omega_1, \Omega_2; \nu)$ is defined as
\[
\rho(\Omega_1, \Omega_2; \nu) = \sup \left\{ \mathsf{Cov}[f, g] : f : \Omega_1 \rightarrow \mathbb{R}, g : \Omega_2 \rightarrow \mathbb{R}, \mathsf{Var}[f] = \mathsf{Var}[g] = 1 \right\}.
\]
Given a distribution $\mu$ on $\Omega_1 \times ... \times \Omega_d$, the correlation $\rho(\Omega_1, ... , \Omega_d; \mu)$ is defined as
\[
\rho(\Omega_1, ... \Omega_d; \mu) = 
\max_{1 \leq i \leq d} \rho \bigg( \prod_{j \neq i} \Omega_j , \Omega_i; \mu \bigg).
\]
\end{definition}
Since our distribution $\mue$ on $\Omega_1 \times ... \times \Omega_d$ with $\Omega_1 = ... = \Omega_d = [d-1]$ is independently noised with probabilty $\delta$, we have the following bound on the correlation. 
\begin{lemma} [Lemma 2.9 of~\cite{Mossel10}] 
    $
    \rho(\Omega_1, ... \Omega_d; \mue) \leq 1 - \delta^d / 2.
    $
\end{lemma}

\begin{definition}
For any function $F : [d-1]^R \to\R$, the  {\em Efron-Stein} decomposition is given by 
\[
F(y) = \sum_{S \subseteq [R]} F_S (y)
\]
where the functions $F_S$ satisfy 
\begin{itemize}
    \item $F_S$ only depends on $y_S$, the restriction of $y$ to the coordinates of $S$.
    \item For all $S \not\subseteq S'$ and all $z_{S'}$, $\E_y[F_S(y) | y_{S'} = z_{S'}] = 0$. 
\end{itemize}
\end{definition}
Based on the Efron-Stein decomposition, we can define (low-degree) influences of a function. For a function $F : [d-1]^R \to \R$ and $p \geq 1$, let $\| F \|_p := \E[|F(y)|^p]^{1/p}$
\begin{definition}\label{def:influence}
For any function $F : [d-1]^R \to \R$,
its {\em $i$th influence} is defined as 
\[
\Inf_i := \sum_{S : i \in S} \| f_S \|_2^2.
\]
Its {\em $i$th degree-$t$ influence} is defined as 
\[
\Inf^{\leq t}_i := \sum_{S : i \in S, |S| \leq t} \| f_S \|_2^2,
\]
\end{definition}

We crucially use the following invariance principle applied to our dictatorship test, using the fact that $\mu$ and $\mue$ are pairwise independent. 

\begin{theorem} [Theorem 1.14 of~\cite{Mossel10}]
For any $\eps > 0$, there exist $t \in \N$ and $\tau > 0$ such that the following is true. Let $F_1, \dots, F_d : [d-1]^R \to [0, 1]$. If $\max(\Inf^{\leq t}_i[F_1], \dots,  \Inf^{\leq t}_i[F_d]) \leq \tau$ for every $i \in [R]$, \[
\E_{(x^1, \dots, x^d) \in \mue^{\otimes R}}[\prod_{j=1}^d F_j (x^d)] \geq 
\prod_{j=1}^d \E[F_j(x^j)]
- \eps.
\]
\label{thm:invariance}
\end{theorem}

Given this theorem, consider an arbitrary set of vertices $S \subseteq \Omega^R$ such that $|S| = \alpha|\Omega^R|$, and let $F_1, \dots, F_d$ all be the reverse indicator function of $S$; $F_i(x) = 1$ if $x \notin S$ and $0$ otherwise. Then the fraction of the hyperedges not intersected by $S$ is 
\[
\E_{(x^1, \dots, x^d) \in \mue^{\otimes R}}[\prod_{j=1}^d F_j (x^j)] \geq 
\prod_{j=1}^d \E[F_j(x^j)]
- \eps
= (1 - \alpha)^d - \eps,
\]
unless $F_i$ has an influential coordinate.

\subsubsection{Reduction from Unique Games}
In this subsection, we introduce the reduction from the Unique Games using the dictatorship test constructed.
We first introduce the Unique Games Conjecture~\cite{Khot02}, which is stated below.
\begin{definition} [Unique Games]
An instance $\mathcal{L} (G(V \cup W, E), [R], \left\{ \pi(v, w) \right\}_{(v, w) \in E})$ of Unique Games consists of a regular bipartite graph $G(V \cup W, E)$ and a set $[R]$ of labels. For each edge $(v, w) \in E$ there is a constraint specified by a permutation $\pi(v, w) : [R] \rightarrow [R]$. Given a labeling $\ell : V \cup W \rightarrow [R]$, let $\Valug(\ell)$ be the fraction of edges satisfied by $\ell$, where an edge $e = (v, w)$ is said to be satisfied if $\ell(v) = \pi(v, w)(\ell(w))$. 
Let $\Optug(\mathcal{L}) = \max_{\ell} (\Valug(\ell))$. 
\end{definition}
\begin{conjecture} [Unique Games Conjecture~\cite{Khot02}] 
\label{conj:ug}
For any constant $\eta > 0$, there is $R = R(\eta)$ such that, for a Unique Games instance $\mathcal{L}$ with label set $[R]$, it is NP-hard to distinguish between
\begin{itemize}
\item $\Optug(\mathcal{L}) \geq 1 - \eta$. 
\item $\Optug(\mathcal{L}) \leq \eta$. 
\end{itemize}
\end{conjecture}

Given $d$ and $\eps$, fix $t, \tau$ from Theorem~\ref{thm:invariance}. 
Given an instance $\mathcal{L} (G(V \cup W, E), [R], \left\{ \pi(v, w) \right\}_{(v, w) \in E})$ of Unique Games, we construct a weighted $d$-uniform hypergraph.
For $y \in [d-1]^R$ and a permutation $\pi : [R] \rightarrow [R]$, let $y \circ \pi \in [d-1]^R$ be defined by $(y \circ \pi)_i = (y)_{\pi^{-1}(i)}$.
\begin{itemize}
\item The set of vertices $\calV := V \times [d-1]^R$.
\item The hyperedges are described by the following probabilistic procedure to sample a $d$-tuple of vertices, where the weight of the hyperedge is exactly the probability that it is sampled. 
\begin{itemize}
\item Sample $w \in W$ uniformly at random and its neighbors $v_1, \dots, v_d$ uniformly and independently. 
\item Sample $(x^1, \dots, x^d) \sim \mue^{\otimes R}$. Output the pair $((v_1, x^1 \circ \pi_{v_1, w}), \dots, (v_d, x^d \circ \pi_{v_d, w}))$. 
\end{itemize}
\end{itemize}

\paragraph{Completeness.} Suppose that $\Valug(\ell) \geq 1 - \eta$ for some labeling $\ell : V \cup W \rightarrow [R]$. 
Consider the set of vertices $\mathcal{U} := \{ (v, x) : x_{\ell(v)} = 1 \}$. In the above sample procedure, the probability that $\ell$ satisfies all the edges $(w, v_1), \dots, (w, v_d)$ is at least $1 - d \eta$, and if that happens, the probability that the sampled hyperedge intersects with $\mathcal{U}$ is at least $1 - \delta$. Therefore, the total fraction of hyperedge intersected by $\mathcal{U}$ is at least $1 - d \eta - \delta$, which can be made at least $1 - \eps$ by choosing small enough $\eta$ and $\delta$. 

\paragraph{Soundness.}
Consider any solution $\calU \subseteq \calV$ with $\calU = \alpha |\calV|$ and suppose that it intersects at least a $1 - (1 - \alpha)^d + \eps$ fraction of the hyperedges. 

For any $v \in V$, let $F_{v} : [d - 1]^R \to \{ 0, 1 \}$ be the reverse indicator function 
$\calU$ for $v$; $F_v(x) = 0$ if $(v, x) \in \calU$ and $1$ otherwise. For each $w \in W$, let $F_{w} : [d - 1]^R \to [0, 1]$ be such that 
\[
F_{w}(y) := \E_{(v, w) \in E} [F_{v}(y \circ \pi_{v, w})].
\]

Then the soundness of the dictatorship test shows that, for some $\tau$ and $t$ depending only on $d$ and $\eps$, unless there exists $i \in [R]$ with $\Inf^{\leq d}_i[F_{w}] > \tau$, the probability that a random hyperedge, in the above probabilistic procedure given $w$, intersects $\calU$ is at most $1 - \E[F_w]^d + \eps / 2$.

Call $w$ {\em good} if $w$ has $i \in [R]$ with $\Inf^{\leq d}_i[F_{w, u}] > \tau$. We define a labeling of the Unique Games instance as follows. First, for any good $w$, let $\ell(w) = i$. Let $\beta$ be the fraction of good $w$'s. Since the fraction of the hyperedges intersected by $\calU$ is at most
\begin{align*}
& \E_{w \mbox{ bad}}
[1 - \E[F_w]^d + \eps / 2] + \beta
\leq 
1 - (\E_{w \mbox { bad}}\E[F_w])^d + \eps / 2  + \beta \\
\leq \, & 1 - (1 - \min(\alpha + \beta, 1))^d + \eps/2 + \beta \leq 
1 - (1 - \alpha)^d + \eps/2 + (d + 1)\beta.
\end{align*}
(The first inequality uses the fact that $1 - (1-x)^d$ is concave in $x$ in $[0, 1]$.) Since this fraction is at least $1 - (1 - \alpha)^d + \eps$, we have $\beta > \frac{\eps}{4d}$. 

From the representation of influences in terms of Fourier coefficients (see Definition~\ref{def:influence}), 
\[
\tau < \Inf_i^{\leq d}[F_{w}] \leq \E_{(v, w) \in E}[\Inf_{\pi_{v, w}(i)}^{\leq d}[F_{v}]]
\]
and we conclude that a $\tau / 2$ fraction of neighbors $v$ of $w$ have $\Inf_{\pi_{v, w}(i)}^{\leq d}(F_{v}) \geq \tau / 2$. 
We choose $\ell(v)$ uniformly from a candidate set $\{ i : \Inf_{i}^{\leq d}[F_{v}] \geq \tau / 2 \}$. (If $v$ has no candidate, choose $\ell(v)$ arbitrarily.)
Since $\sum_i \Inf_{i}^{\leq d}[F_{v}] \leq d$, there are at most $O(d/\tau)$ possible $i$'s, so the total size of the above set is $O(\frac{d}{\tau})$. 

So this labeling strategy satisfies at least $\Omega(\frac{\beta \tau}{d})$ fraction of Unique Games constraints in expectation. 
Taking $\eta$ small enough completes the proof of Theorem~\ref{thm:hvc}.

\bibliographystyle{plain}
\bibliography{references}

\begin{thebibliography}{10}

\bibitem{adamczyk2019constant}
Marek Adamczyk, Jaroslaw Byrka, Jan Marcinkowski, Syed~M Meesum, and Michal
  Wlodarczyk.
\newblock Constant-factor fpt approximation for capacitated k-median.
\newblock In {\em 27th Annual European Symposium on Algorithms (ESA 2019)}.
  Schloss Dagstuhl-Leibniz-Zentrum fuer Informatik, 2019.

\bibitem{agrawal2023clustering}
Akanksha Agrawal, Tanmay Inamdar, Saket Saurabh, and Jie Xue.
\newblock Clustering what matters: optimal approximation for clustering with
  outliers.
\newblock {\em Journal of Artificial Intelligence Research}, 78:143--166, 2023.

\bibitem{arya2001local}
Vijay Arya, Naveen Garg, Rohit Khandekar, Adam Meyerson, Kamesh Munagala, and
  Vinayaka Pandit.
\newblock Local search heuristic for k-median and facility location problems.
\newblock In {\em Proceedings of the thirty-third annual ACM symposium on
  Theory of computing}, pages 21--29, 2001.

\bibitem{austrin2009approximation}
Per Austrin and Elchanan Mossel.
\newblock Approximation resistant predicates from pairwise independence.
\newblock {\em Computational Complexity}, 18(2):249--271, 2009.

\bibitem{badanidiyuru2012approximating}
Ashwinkumar Badanidiyuru, Robert Kleinberg, and Hooyeon Lee.
\newblock Approximating low-dimensional coverage problems.
\newblock In {\em Proceedings of the twenty-eighth annual symposium on
  Computational geometry}, pages 161--170, 2012.

\bibitem{bandyapadhyay2022fpt}
Sayan Bandyapadhyay, Fedor~V Fomin, Petr~A Golovach, Nidhi Purohit, and Kirill
  Simonov.
\newblock Fpt approximation for fair minimum-load clustering.
\newblock In {\em 17th International Symposium on Parameterized and Exact
  Computation}, 2022.

\bibitem{bandyapadhyay2023fpt}
Sayan Bandyapadhyay, William Lochet, and Saket Saurabh.
\newblock Fpt constant-approximations for capacitated clustering to minimize
  the sum of cluster radii.
\newblock In {\em 39th International Symposium on Computational Geometry (SoCG
  2023)}. Schloss Dagstuhl-Leibniz-Zentrum f{\"u}r Informatik, 2023.

\bibitem{bartal1998approximating}
Yair Bartal.
\newblock On approximating arbitrary metrices by tree metrics.
\newblock In {\em Proceedings of the thirtieth annual ACM symposium on Theory
  of computing}, pages 161--168, 1998.

\bibitem{byrka2017improved}
Jaros{\l}aw Byrka, Thomas Pensyl, Bartosz Rybicki, Aravind Srinivasan, and Khoa
  Trinh.
\newblock An improved approximation for k-median and positive correlation in
  budgeted optimization.
\newblock {\em ACM Transactions on Algorithms (TALG)}, 13(2):1--31, 2017.

\bibitem{charikar1998rounding}
Moses Charikar, Chandra Chekuri, Ashish Goel, and Sudipto Guha.
\newblock Rounding via trees: deterministic approximation algorithms for group
  steiner trees and k-median.
\newblock In {\em Proceedings of the thirtieth annual ACM symposium on Theory
  of computing}, pages 114--123, 1998.

\bibitem{charikar1999improved}
Moses Charikar and Sudipto Guha.
\newblock Improved combinatorial algorithms for the facility location and
  k-median problems.
\newblock In {\em 40th Annual Symposium on Foundations of Computer Science
  (Cat. No. 99CB37039)}, pages 378--388. IEEE, 1999.

\bibitem{charikar1999constant}
Moses Charikar, Sudipto Guha, {\'E}va Tardos, and David~B Shmoys.
\newblock A constant-factor approximation algorithm for the k-median problem.
\newblock In {\em Proceedings of the thirty-first annual ACM symposium on
  Theory of computing}, pages 1--10, 1999.

\bibitem{charikar2012dependent}
Moses Charikar and Shi Li.
\newblock A dependent lp-rounding approach for the k-median problem.
\newblock In {\em International Colloquium on Automata, Languages, and
  Programming}, pages 194--205. Springer, 2012.

\bibitem{charikar2009integrality}
Moses Charikar, Konstantin Makarychev, and Yury Makarychev.
\newblock Integrality gaps for sherali-adams relaxations.
\newblock In {\em Proceedings of the forty-first annual ACM symposium on Theory
  of computing}, pages 283--292, 2009.

\bibitem{cohen2019tight}
Vincent Cohen-Addad, Anupam Gupta, Amit Kumar, Euiwoong Lee, and Jason Li.
\newblock Tight fpt approximations for $ k $-median and $ k $-means.
\newblock {\em arXiv preprint arXiv:1904.12334}, 2019.

\bibitem{cgkll19}
Vincent Cohen-Addad, Anupam Gupta, Amit Kumar, Euiwoong Lee, and Jason Li.
\newblock Tight fpt approximations for k-median and k-means.
\newblock In {\em 46th International Colloquium on Automata, Languages, and
  Programming (ICALP 2019)}, volume 132, pages 42--1. Schloss
  Dagstuhl--Leibniz-Zentrum fuer Informatik, 2019.

\bibitem{cohen2019fixed}
Vincent Cohen-Addad and Jason Li.
\newblock On the fixed-parameter tractability of capacitated clustering.
\newblock In {\em 46th International Colloquium on Automata, Languages, and
  Programming (ICALP 2019)}, volume 132, pages 41--1. Schloss
  Dagstuhl--Leibniz-Zentrum fuer Informatik, 2019.

\bibitem{cohen2023breaching}
Vincent Cohen-Addad~Viallat, Fabrizio Grandoni, Euiwoong Lee, and Chris
  Schwiegelshohn.
\newblock Breaching the 2 lmp approximation barrier for facility location with
  applications to k-median.
\newblock In {\em Proceedings of the 2023 Annual ACM-SIAM Symposium on Discrete
  Algorithms (SODA)}, pages 940--986. SIAM, 2023.

\bibitem{feige1998threshold}
Uriel Feige.
\newblock A threshold of ln n for approximating set cover.
\newblock {\em Journal of the ACM (JACM)}, 45(4):634--652, 1998.

\bibitem{fl11}
Dan Feldman and Michael Langberg.
\newblock A unified framework for approximating and clustering data.
\newblock In {\em Proceedings of the forty-third annual ACM symposium on Theory
  of computing}, pages 569--578, 2011.

\bibitem{feng2020unified}
Qilong Feng, Zhen Zhang, Ziyun Huang, Jinhui Xu, and Jianxin Wang.
\newblock A unified framework of fpt approximation algorithms for clustering
  problems.
\newblock In {\em 31st International Symposium on Algorithms and Computation
  (ISAAC 2020)}. Schloss Dagstuhl-Leibniz-Zentrum f{\"u}r Informatik, 2020.

\bibitem{ghazi2017lp}
Badih Ghazi and Euiwoong Lee.
\newblock Lp/sdp hierarchy lower bounds for decoding random ldpc codes.
\newblock {\em IEEE Transactions on Information Theory}, 64(6):4423--4437,
  2017.

\bibitem{gowda2023improved}
Kishen~N Gowda, Thomas Pensyl, Aravind Srinivasan, and Khoa Trinh.
\newblock Improved bi-point rounding algorithms and a golden barrier for
  k-median.
\newblock In {\em Proceedings of the 2023 Annual ACM-SIAM Symposium on Discrete
  Algorithms (SODA)}, pages 987--1011. SIAM, 2023.

\bibitem{goyal2023tight}
Dishant Goyal and Ragesh Jaiswal.
\newblock Tight fpt approximation for socially fair clustering.
\newblock {\em Information Processing Letters}, 182:106383, 2023.

\bibitem{goyal2020fpt}
Dishant Goyal, Ragesh Jaiswal, and Amit Kumar.
\newblock Fpt approximation for constrained metric $ k $-median/means.
\newblock {\em arXiv preprint arXiv:2007.11773}, 2020.

\bibitem{guha1999greedy}
Sudipto Guha and Samir Khuller.
\newblock Greedy strikes back: Improved facility location algorithms.
\newblock {\em Journal of algorithms}, 31(1):228--248, 1999.

\bibitem{haastad2001some}
Johan H{\aa}stad.
\newblock Some optimal inapproximability results.
\newblock {\em Journal of the ACM (JACM)}, 48(4):798--859, 2001.

\bibitem{inamdar2020capacitated}
Tanmay Inamdar and Kasturi Varadarajan.
\newblock Capacitated sum-of-radii clustering: An fpt approximation.
\newblock In {\em 28th Annual European Symposium on Algorithms (ESA 2020)}.
  Schloss Dagstuhl-Leibniz-Zentrum f{\"u}r Informatik, 2020.

\bibitem{jain2003greedy}
Kamal Jain, Mohammad Mahdian, Evangelos Markakis, Amin Saberi, and Vijay~V
  Vazirani.
\newblock Greedy facility location algorithms analyzed using dual fitting with
  factor-revealing lp.
\newblock {\em Journal of the ACM (JACM)}, 50(6):795--824, 2003.

\bibitem{jain2001approximation}
Kamal Jain and Vijay~V Vazirani.
\newblock Approximation algorithms for metric facility location and k-median
  problems using the primal-dual schema and lagrangian relaxation.
\newblock {\em Journal of the ACM (JACM)}, 48(2):274--296, 2001.

\bibitem{jain2023parameterized}
Pallavi Jain, Lawqueen Kanesh, Fahad Panolan, Souvik Saha, Abhishek Sahu, Saket
  Saurabh, and Anannya Upasana.
\newblock Parameterized approximation scheme for biclique-free max k-weight sat
  and max coverage.
\newblock In {\em Proceedings of the 2023 Annual ACM-SIAM Symposium on Discrete
  Algorithms (SODA)}, pages 3713--3733. SIAM, 2023.

\bibitem{jaiswal2023clustering}
Ragesh Jaiswal and Amit Kumar.
\newblock Clustering what matters in constrained settings.
\newblock {\em arXiv preprint arXiv:2305.00175}, 2023.

\bibitem{Khot02}
Subhash Khot.
\newblock On the power of unique 2-prover 1-round games.
\newblock In {\em Proceedings of the thiry-fourth annual ACM symposium on
  Theory of computing}, pages 767--775, 2002.

\bibitem{khot2007optimal}
Subhash Khot, Guy Kindler, Elchanan Mossel, and Ryan O’Donnell.
\newblock Optimal inapproximability results for max-cut and other 2-variable
  csps?
\newblock {\em SIAM Journal on Computing}, 37(1):319--357, 2007.

\bibitem{li2013approximating}
Shi Li and Ola Svensson.
\newblock Approximating k-median via pseudo-approximation.
\newblock In {\em proceedings of the forty-fifth annual ACM symposium on theory
  of computing}, pages 901--910, 2013.

\bibitem{lin1992approximation}
Jyh-Han Lin and Jeffrey~Scott Vitter.
\newblock Approximation algorithms for geometric median problems.
\newblock {\em Information Processing Letters}, 44(5):245--249, 1992.

\bibitem{lin1992approximations}
Jyh-Han Lin and Jeffrey~Scott Vitter.
\newblock e-approximations with minimum packing constraint violation.
\newblock In {\em Proceedings of the twenty-fourth annual ACM symposium on
  Theory of computing}, pages 771--782, 1992.

\bibitem{manurangsi2019note}
Pasin Manurangsi.
\newblock A note on max k-vertex cover: Faster fpt-as, smaller approximate
  kernel and improved approximation.
\newblock In {\em 2nd Symposium on Simplicity in Algorithms}, 2019.

\bibitem{Mossel10}
Elchanan Mossel.
\newblock Gaussian bounds for noise correlation of functions.
\newblock {\em Geometric and Functional Analysis}, 19(6):1713--1756, 2010.

\bibitem{xu2020constant}
Yicheng Xu, Rolf~H M{\"o}hring, Dachuan Xu, Yong Zhang, and Yifei Zou.
\newblock A constant fpt approximation algorithm for hard-capacitated k-means.
\newblock {\em Optimization and Engineering}, 21:709--722, 2020.

\end{thebibliography}

\end{document}